\documentclass{article}

\usepackage[preprint]{neurips_2021}

\usepackage[utf8]{inputenc} % allow utf-8 input
\usepackage[T1]{fontenc}    % use 8-bit T1 fonts
\usepackage{hyperref}       % hyperlinks
\usepackage{url}            % simple URL typesetting
\usepackage{booktabs}       % professional-quality tables
\usepackage{amsfonts}       % blackboard math symbols
\usepackage{nicefrac}       % compact symbols for 1/2, etc.
\usepackage{microtype}      % microtypography
\usepackage{xcolor}         % colors
\usepackage{amsmath}
\usepackage{dsfont}
\usepackage{amsthm}
\usepackage[linesnumbered, ruled,vlined]{algorithm2e}
\usepackage{graphicx}
\usepackage[subtle]{savetrees}
\usepackage{centernot}
\usepackage{multirow}
\usepackage{enumitem}
\usepackage{svg}

\newcommand{\CI}{\mathrel{\perp\mspace{-10mu}\perp}}
\newcommand{\nCI}{\centernot{\CI}}

\newtheorem{theorem}{Theorem}
\newtheorem{lemma}{Lemma}
\newtheorem{assumption}{Assumption}
\newtheorem{definition}{Definition}

\newtheorem{remark}{Remark}

\title{Auctions and Peer Prediction \\ for Academic Peer Review}

\author{%
  Siddarth Srinivasan \\
  School of Computer Science and Engineering\\
  University of Washington\\
  Seattle, WA 98195 \\
  \texttt{sidsrini@cs.washington.edu} \\
  \And
  Jamie Morgenstern \\
  School of Computer Science and Engineering\\
  University of Washington\\
  Seattle, WA 98195 \\
  \texttt{jamiemmt@cs.washington.edu} \\
}

\begin{document}

\maketitle

\begin{abstract}
Peer reviewed publications are considered the gold standard in certifying and disseminating ideas that a research community considers valuable. However, we identify two major drawbacks of the current system: (1) the overwhelming demand for reviewers due to a large volume of submissions, and (2) the lack of incentives for reviewers to participate and expend the necessary effort to provide high-quality reviews. In this work, we adopt a mechanism-design approach to propose improvements to the peer review process, tying together the paper submission and review processes and simultaneously incentivizing high-quality submissions and reviews. In the submission stage, authors participate in a VCG auction for review slots by submitting their papers along with a bid that represents their expected value for having their paper reviewed. For the reviewing stage, we propose a novel peer prediction mechanism (H-DIPP) building on recent work in the information elicitation literature, which incentivizes participating reviewers to provide honest and effortful reviews. The revenue raised in the submission stage auction is used to pay reviewers based on the quality of their reviews in the reviewing stage.
\end{abstract}

\section{Introduction}

Scientific publishing through peer review is widely regarded as the gold standard for disseminating scientific research. Yet, despite its widespread use and acceptance, it is plagued by a number of well-documented problems: willing reviewers are often in short supply and reviews themselves are often of poor quality, giving rise to noisy acceptance decisions \citep{lawrence_2015, shah2018design}. Studying the history of peer review in sociology, \citet{merriman2020peer} writes that while peer review is `a practice intended to assess the merits of pieces of scholarship, its features did not arise because they were believed to be an especially apt means of identifying the best work.' The issues with peer review have become especially salient in machine learning as well as computing more broadly. Computing relies largely on a conference model for publishing and disseminating research, and while the number of submissions to major conferences has exploded, the pool of qualified reviewers has grown more slowly \citep{sculley2018winner, shah2019principled, stelmakh2020novice}. As these issues are only likely to persist, a solution addressing the problem of low-quality reviews is needed.

In this work, we adopt a mechanism-design approach to tackle two drawbacks of the current system of conference-based publishing in machine learning: (1) a large number of submissions, leading to poor and often arbitrary criteria for desk rejection, and (2) poor incentives for reviewers that make it challenging to find qualified reviewers willing to provide effortful and honest reviews. Our proposed mechanism tackles (1) by requiring authors to participate in an auction for review slots, and tackles (2) via a peer prediction mechanism (involving predicting peers' review scores) that rewards reviewers for honest and effortful reviews. The mechanism involves minimal modification to the current peer review system; it requires authors to submit a bid accompanying each paper, and reviewers to participate in a simple prediction game after submitting their reviews. The paper-reviewer matching process and accept/reject decision process may be independently improved or remain unchanged. Our mechanism would adopt a novel currency usable within a consortium of conferences, with the goal of incentivizing authors to provide the much needed service of peer review in order to earn credits exchangeable for having their own papers reviewed. We outline our proposed process below:
\begin{enumerate}
    \item \textbf{Reviewer Signup}: The conference announces a reviewer signup period open until a few days prior to the submission deadline.
    \item \textbf{Review Slot Announcement}: Based on the number of reviewers signed up, the conference estimates and announces the number of review slots $P$ that can be supported without overburdening reviewers and compromising review quality.
    \item \textbf{Review Slot Auction}: Each submitted paper is accompanied by an author's bid for a review slot. A VCG auction determines which papers get review slots and how much the bidder pays. Authors are not charged for papers that don't receive review slots.
    \item \textbf{Reviewer Matching} (\emph{unchanged}): The conference assigns reviewers to papers accepted for review per the usual approach, ensuring at least 3 well-matched reviewers to a paper.
    \item \textbf{Reviewing}: Each reviewer scores their assigned paper on several criteria and submits their review per the usual process. \emph{Additionally}, they play a simple asynchronous game after all reviews are submitted, involving predicting another reviewer's score. Reviewers are compensated in artificial currency according to their performance in this game.
    \item \textbf{Accept/Reject Decisions} (\emph{unchanged}): Decisions are made  per the standard process.
\end{enumerate}

While our proposal is generally applicable to peer review, our perspective is informed by the conference publication model in machine learning. There remain many questions to be answered before our proposed mechanism can be deployed in practice, and we primarily intend for this work to lay the theoretical foundations for a plausible (but by no means comprehensive) approach to be iterated upon in designing a robust scientific publication process. Our paper is structured as follows: in section 2, we provide some background on academic peer review and motivate the problem; in section 3, we describe our proposed mechanism; in section 4, we present related work; and in Section 5, we conclude with a discussion of some important considerations. 

\paragraph{Main Contributions} We make two main contributions: (1) we propose a peer review mechanism consisting of a submission stage and a reviewing stage, where the revenue raised by the former is used to pay for the latter, and (2) we design the reviewing stage by building on previous work in the information elicitation literature to develop a novel peer prediction mechanism that incentivizes reviewers to provide honest and effortful reviews.

\begin{figure}
    \centering
    \includegraphics[width=\textwidth]{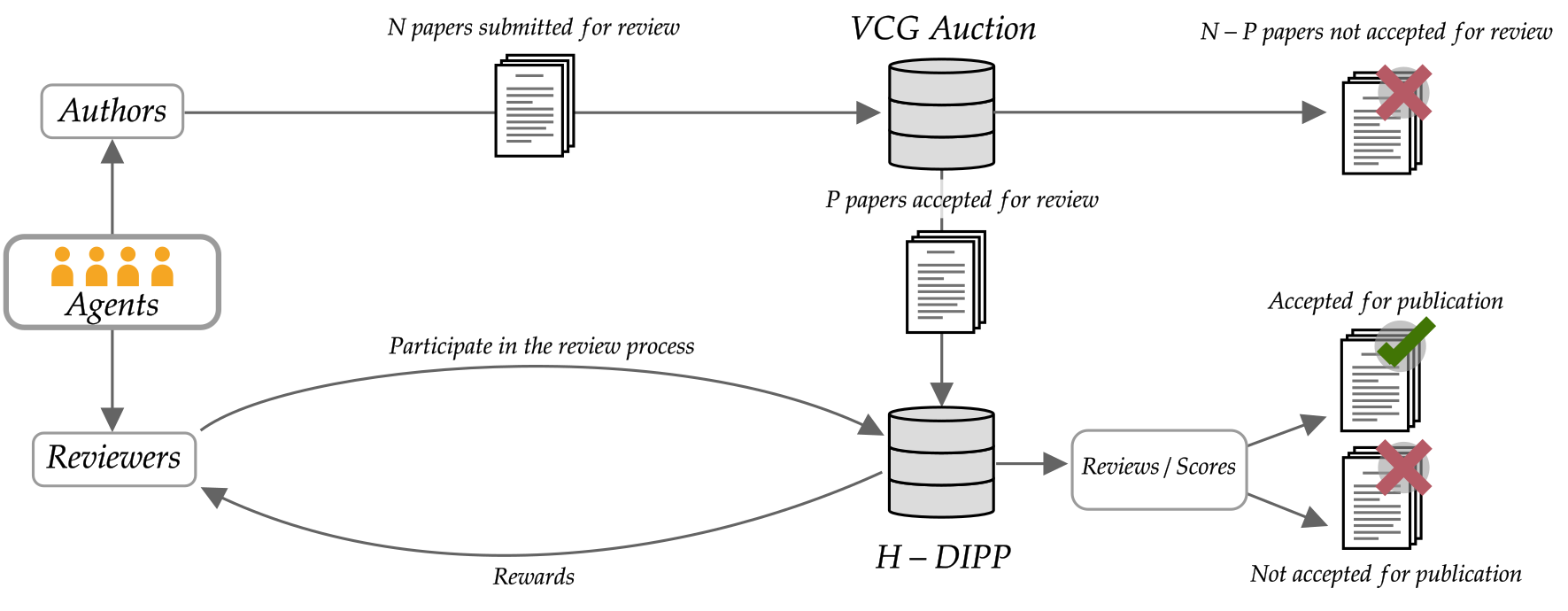}
    \caption{Diagram of proposed mechanism for academic peer review}
    \label{fig:my_label}
\end{figure}

\section{Background}

Although there are some minor variations, the process of peer review can be summarized as follows: (1) authors submit their work to a conference, (2) some papers receive desk rejection, (3) a paper that is not desk rejected is matched with a small number (3-5) of willing reviewers (who do not know the identity of the authors and vice-versa), (4) the reviewers independently score the paper's quality for publication  based on several criteria such as technical correctness, interest to the community, and quality of exposition (the reviewers may or may not know each others' identities), (5) a meta-reviewer aggregates the reviewers' scores and makes an accept/reject decision, incorporating the conference's targeted acceptance rate into their decision.  With this context, we now discuss the often cited primary motivations for and criticisms of peer review below. \citet{merriman2020peer} argues that these `intellectual' motivations tend to be post-hoc rationalizations, and the current process of peer review is `best understood as the product of continuous efforts to steward editors’ scarce attention while preserving an open submission policy that favors authors’ interests'. Part of our goal then is to develop a peer review mechanism that lives up to the stated `intellectual' motivations for the process.

\subsection{Motivations for Peer Review}\label{motivs}

\begin{enumerate}
\item[\bf M1] {\bf Screening Results for Validity}: Peer review serves as a community-issued certificate that the approach and results in a publication seem plausible, replicable, and not fraudulent.

\item[\bf M2] {\bf Screening Results for Importance to the Community}:  The publication record serves as a central repository of curated results that the field considers important or impactful and wishes to focus on; the conference itself serves to build consensus and host debates on the cutting-edge issues in the field.

\item[\bf M3] {\bf Reputation Building/Signaling}: Publication at selective venues also serves to build reputation and advance careers of authors; the idea is that a selective acceptance process that identifies valid and interesting results also identifies individuals who do good work.

\item[\bf M4] {\bf Feedback and Networking}: Peer review serves as a systematic approach to obtaining feedback on both ongoing and completed work, as well as finding collaborators interested in tackling similar problems.
\end{enumerate}

M1 and M2 are often cited as the primary motivations for peer review, and are arguably the most foundational. While M3 is not often cited explicitly as a motivation, it is well-known that publishing frequently at top-tier venues improves career advancement prospects; peer review clearly serves as a mechanism to identify individuals who should be rewarded. M3 derives its legitimacy from M1 and M2. Meanwhile, M4 is an auxiliary benefit derived from M1 and M2. We argue that the peer review system should be judged on its ability to satisfy M1 and M2 while cognizant of its social functions M3 and M4 so as to avoid unintended consequences.
 
\subsection{Criticisms of Peer Review}\label{crit}

The general criticism of peer review is that the reviews are often sub-par and the accept/reject decisions are too noisy. Various NeurIPS experiments \citep{lawrence_2015, shah2018design, neural} provide some empirical basis for this criticism; the 2014 NeurIPS experiment in particular (see Appendix \ref{2014exp}) found that if the conference was re-run, about a quarter of papers would have their accept/reject decisions swapped, and only about half of accepted papers would be re-accepted. There is also a perception that reviewers are often misled by unsubstantiated claims \citep{lipton2018troubling} or work that \emph{looks} important but isn't \citep{vazire2017quality}. This may incentivize the submission of premature, lower quality work which if rejected, may be resubmitted at no cost to future conferences with minor tweaks with the goal of getting through a noisy review process \citep{bengio_2020}.  In this context, we discuss some of the explanations for why the current system of peer review within machine learning may be sub-par, although our discussion may apply to peer review in general.

\begin{enumerate}
\item[\bf C1] {\bf The Extensive Margin of Reviewing}: There are too many submissions and not enough reviewers, so the current system struggles to match submissions with qualified and interested reviewers. There are no explicit incentives for reviewers, so reviewers participate entirely for pro-social reasons. The shortage of reviewers has been exacerbated by the recent explosion in submissions to top conference venues in machine learning over the past decade. The large number of submissions can also lead to desk-rejection policies that can be somewhat opaque and arbitrary. Poor matching due to a scarcity of reviewers could explain why reviews are often sub-par, so a mechanism that incentivizes reviewers to participate could be valuable.

    \item[\bf  C2] {\bf The Intensive Margin of Reviewing}: Even when reviewers participate voluntarily, they may \emph{not be willing to exert enough effort} to provide a high-quality assessment of the paper, and there is currently no explicit incentive to do so. In an attempt to redress the shortage of reviewers, some conferences require authors to also participate in review. While this would increase the supply of reviewers, it could potentially crowd-out pro-social motivations and worsen the average quality of reviews without proper incentives. Thus, we may need to provide the appropriate incentives to elicit effortful reviews, especially when requiring participation in the review process.
    
     \item[\bf C3] {\bf Dishonest Reviews}: Some have noted that even when reviewers exert enough effort to understand a paper, they may not report their honest opinion for various reasons. Thus, it is desirable to design the review process to remove obstacles to truthful reviewing.
    
    \item[\bf C4] {\bf Ability of Reviewers}: Another strain of criticism is that peer review is fundamentally flawed and \emph{it cannot work} even when reviewers provide honest and effortful reviews on well-matched papers.  There are two versions of this criticism: the first is that reviewing is too subjective for reviewers to
    agree on results' validity and importance to the community, and the second is that even when reviewers agree, their reviews do not track any meaningful measure of importance to the community. Empirical evidence from the 2014 NeurIPS experiment (where all the aforementioned incentive issues persist) suggests that reviewers' scores correlate moderately well with each other as well as with log-citation counts; it may be quite possible to improve upon these results by remedying incentive issues.
\end{enumerate}

Our focus in this paper is to address the first three criticisms head on with a mechanism-design approach. On the paper submission side of the process, our proposed mechanism disincentivizes lower-quality submissions aimed at getting through a noisy review process. On the paper reviewing side, our proposed mechanism incentivizes reviewers to participate in the peer review process, exert effort, and report their assessments truthfully. Our mechanism is agnostic to the reviewer-paper matching process, and we discuss this in Appendix \ref{demsup}. %The mechanism will also reveal whether the peer review process yielded meaningful reviews by providing a measure of agreement among honest and effortful reviews.
The key challenge we face is the lack of a ground-truth verification for the submitted reviews, which makes it difficult to verify whether reviewers' assessments are truthful and effortful. We tackle this issue by building on recently proposed mechanisms from the information elicitation without verification (IEWV) literature.

\section{A Mechanism for Academic Peer Review}

We now describe our peer review mechanism that incentivizes  authors to submit papers they believe are likely to be accepted, and reviewers to participate, exert effort, and honestly review papers. Stage I deals with the paper submission process, where authors bid for review slots, and Stage II deals with the paper review process, where reviewers are paid based on the quality of their reviews. The paper-reviewer matching process and accept/reject decision process may remain unchanged. We envision the mechanism using its own currency (henceforth referred to simply as `credits') that can be utilized within a consortium of conferences.

\paragraph{Summary} In the Stage I submission mechanism, authors participate in a VCG auction for review slots,
with one author on each paper $\omega_i$ also submitting an accompanying bid $b_i$ (in mechanism-specific credits). The winners of the auction pay the highest losing bid and have their papers  \emph{accepted for review}. In the Stage II reviewing mechanism (which we call H-DIPP), every paper $\omega_i$ accepted for review is assigned at least 3 reviewers. Reviewers proceed  per the usual peer review process, supplying reviews and scoring the paper on a common, fixed set of $T$ criteria (e.g. novelty, correctness, interest to the community, quality of writing). The reviewers will then also  participate in a simple peer prediction game to report \emph{predictions} of the scores of one other peer, twice for each criterion; reviewers first predict the peer's review score on a given criterion, are shown the third peer's score on that criterion, and are asked to update their prediction. This proceeds sequentially for each criterion in a  pre-determined order. Reviewers' predictions of a peer's review scores are rewarded based on their accuracy (with the log proper scoring rule), and reviewers' review scores on each criterion are rewarded based on how much it aids another peer in their prediction game.  We prove that under such a reward scheme, reviewers maximize their expected payoff when they provide honest and effortful reviews (when everyone else also does so). We emphasize that the prediction game involves predicting other reviewers' review scores, and not something like the paper's expected impact. The revenue raised from the submission stage auction can be used to pay reviewers for their reviews in the reviewing stage. %, and the mechanism will seek to accept as many papers for review as it can find well-matched reviewers for within this budget constraint. 
Decisions on which papers are \emph{accepted for publication} can be made per the usual process.  Authors who submitted bids must ensure that they have (or can earn) enough credits to cover their auction bids, or risk having their reviews and potential acceptance withheld. In general, authors face the budget constraint of submitting bids no larger than the value of credits they are willing and able to earn through reviewing. Authors who provide the valuable and scarce service of reviewing will be able to earn and bid more credits, providing a positive incentive to participate in the review process.

\begin{remark}
We note that there are two simple ways to incentivize authors to earn sufficient credits via reviewing to cover their auction bid. The first is to refuse to accept papers through the auction without sufficient store of reviewing credit; the second is to  model the conference submission process as a repeated game where an author can either save credits or be indebted credits from round to round. We do not assume that every author is a reviewer for the same conference, or vice-versa.
\end{remark}

\subsection{Setup}

\paragraph{Notation} Suppose we have $N$ papers submitted to the conference peer review mechanism denoted $\omega_i$ (for $i \in [N]$ ) with $K$ participating agents. We use the index $i$ when referring to papers in general and the index $k$ in referring to agents in general. We represent the papers submitted by agent $k$ as $\{\omega_{i_s}^{k}\}$ and the papers reviewed by an agent $k$ as $\{\omega_{i_r}^{k}\}$ with the natural condition that $\{\omega_{i_s}^{k}\} \cap \{\omega_{i_r}^{k}\} = \emptyset$. Author $k$'s Stage I bid accompanying paper $\omega_i$ is $b_i$. Reviewer $k$'s scores of paper $\omega_i$ on the $T$ criteria are written as $({\bf x}_k^{[T]})_i$, their first predictions of peer $k'$'s reported scores as $(\hat{\bf p}_{k'^{[T]} \leftarrow k})_i$, and their second predictions of peer $k'$'s reported scores as $(\hat{\bf p}^+_{k'^{[T]} \leftarrow k}
)_i$. When reviewer $k$ exerts enough effort to receive informed private signals on $t_k$ of $T$ criteria, the cost of the effort is $e_k(t_k)$. A complete table of notation is presented in Appendix \ref{notapp}.

\paragraph{Mechanism Design Background} We define an agent $k$'s \emph{strategy} as a map from their private information to a distribution over reports elicited by a mechanism. An agent's \emph{submission strategy} consists of their bids for all their submitted papers, denoted $s_k^{(1)} = \left(b_i \right)_{i \in \{\omega_{i_s}^{k}\}}$. An agent's \emph{reviewing strategy} consists of the effort  they exerted to arrive at an informed review score on $t_k$ criteria, reported scores on all $T$ criteria and their two predictions of a peer $k'$s score on every criterion for every paper they review, denoted $s_k^{(2)} = \left( \left\{e_k(t_k), \hat{\bf x}_k^{[T]}, \hat{\bf p}_{k'^{[T]} \leftarrow k}, \hat{\bf p}_{k'^{[T]} \leftarrow k}\right\}_{i}\right)_{i\in \{\omega_{i_r}^{k}\}}$. %The complete strategy for an agent participating in the mechanism is simply $s_k = s_k^{(1)} \times s_k^{(2)}$. 
A \emph{strategy profile} is a tuple of every agent's strategy ${\bf s} = (s_1, s_2, \ldots, s_K)$. The strategy profile of all agents excluding agent $k$ is ${\bf s}_{-k}$.

\begin{definition}[Dominant Strategy]
A strategy $s_k$ is dominant if the strategy profile $(s_k, {\bf s}_{-k})$ maximizes agent $k$'s payoff for any strategy profile ${\bf s}_{-k}$.
\end{definition}

\begin{definition}[(Strict) Bayesian Nash Equilibrium (BNE)]
Suppose player $k$ observes private signal $x_k$. A strategy profile ${\bf s} = (s_1(x_1), s_2(x_2), \ldots, s_K(x_K))$ is a Bayesian Nash Equilibrium if no agent $k$ can increase their expected payoff by deviating from strategy $s_k(\cdot)$ when the other agents play ${\bf s}_{-k}$, for any private signal $x_k$. The BNE is strict if deviating from $s_k$ decreases agent $k$'s expected payoff.  If the strategy profile of the BNE consists of identical strategies for all agents, the BNE is symmetric.
\end{definition}

\paragraph{Utility Function} We specify the utility function $U^{(1)}_k$ of an author agent $k$ participating in the submission mechanism (Stage I)  as follows:
\begin{equation}\label{eq:util1}
    U_k^{(1)} = \begin{cases} \sum_{i \in \{\omega_{i_s}^{k}\}} (v_k\delta_i - c_i) & \text{accepted for review} \\ 0 & \text{reject without  review} \end{cases}
\end{equation}
where $\delta_i$ is a random variable indicating whether or not the paper is accepted \emph{for publication}, $v_k$ is the agent's value in credits for a paper accepted for publication and $c_i$ is the price charged by the VCG auction (Stage I) for a review slot.
We can similarly specify the utility function $U_k^{(2)}$ of a reviewing agent $k$ participating in the H-DIPP reviewing mechanism (Stage II) as follows:
\begin{equation}\label{eq:exputil2}
    U_k^{(2)} = \sum_{i\in \{\omega_{i_r}^{k}\}}  r_k\left( \left[s_k^{(2)}\right]_i, \left[{\bf s}_{-k}^{(2)}\right]_i\right) - \left({e}_{k}(t_k)\right)_i 
\end{equation}
where $r_k$ is the reward credits paid to reviewer $k$ by the H-DIPP mechanism for a given strategy $[s_k^{(2)}]_i$ on paper $\omega_i$. % The total utility of an agent participating in the mechanism is $U_k = U_k^{(1)} + U_k^{(2)}$.  
% We design the mechanism so it does not impact the ex ante quality of review ie likelihood of acceptance
We analyze the incentives and strategy profiles of the two stages separately. % This would be appropriate if the value of a review slot and payments to reviewers had no bearing on each other. We discuss when this may fail and what can be done about it in Section \ref{demsup}.
With this setup, we now present the submission and reviewing stages of the mechanism.

% Fix all other agents' papers, bids, decision to review. If *I* bid higher in Stage I, if I'm the marginal author I can increase the revenue to the mechanism. This extra revenue can only increase the number of accepted papers. Accepting additional papers could mean adding reviewers who might be more expensive. This means I can increase my payment rate. SO: to the extent a 1c higher bid increases the payment by more than 1c, it makes it worth it. Similarly, if I choose to participate, I increase the supply of reviewers, which can end up either increasing number of review slots while decreasing average review quality or decreasing number of review slots and increasing review quality, both of which can change the prob of acceptance and hence a review slot, and hence the bid. Once you participate though, you wanna play it right.

% Strange bc ex ante homogeneous slots, even tho that's obv not the case...

% We note that the utilities and strategies of the two stages are independent, in that the strategy in one stage has no bearing on the utility achieved in the other stage. 

\subsection{Stage I: Paper Submission} \label{stage1}
Our Stage I mechanism  uses a VCG auction to allocate review slots and charge winners the threshold bid necessary for their paper to be reviewed. We assume that the number of review slots $P$ is fixed and announced by the conference before paper submission deadline, based on how many reviews it expects to be able to pay for (perhaps with some predictive model using historical data, with some leeway to run a deficit). We defer a discussion of how the number of review slots $P$ may be determined dynamically to Appendix \ref{demsup}, along with the associated considerations and challenges. The auction asks authors to submit a bid $b_i$ along with every paper $\omega_i$ indicating how much they are willing to pay for a review slot for that paper. %If an author truly values a paper being accepted for publication at $b_i$ credits, it can equivalently be stated as the willingness to spend the necessary time and effort to earn $b_i$ credits as a reviewer.
Authors can observe the quality of their paper $\omega_i$ and arrive at a belief that the paper will be accepted for publication with probability $\eta_i$ conditional on review under the equilibrium of the submission mechanism\footnote{In other words, authors should reason about their paper's acceptance probability in a pool with a particular combination of $P$ of the `the best papers' and papers from prolific reviewers who can bid higher since they review more. This would be the equilibrium of the submission mechanism; if the pool had some other combination of papers, a stronger paper or a prolific reviewer with larger store of credits would be willing to bid higher to displace a paper in the pool with their paper. We are assuming that by attending conferences and participating in review, authors have some notion of how competitive their paper will be in relative terms.} (\emph{bids} have no bearing on the chance of acceptance for publication). Hence, denoting reviewer $k$'s value of a paper accepted for publication as $v_k$, the expected value of a review slot is $\mathds{E}[v_k \delta_i] = v_k\eta_i$.
We also have the following natural conditions: (1) the review slots are identical\footnote{Even though from the reviewing side, the review slots are heterogeneous (papers will differ in terms of cost of finding well-matched reviewers), we treat the slots as identical from the author's side in the interest of fairness. This could be relaxed by splitting the conference into `tracks' within which this treatment is more justifiable.}, (2) each paper may only receive a single review slot, and (3) package bids are not allowed, i.e., authors cannot make a single bid for multiple review slots for multiple papers, but must instead make a separate bid for each slot for each paper. In this setting, it is easy and practical to identify the winners of the auction when a single author can submit multiple papers for review. %; the concern they address is that authors might try to strategically bid on their multiple submissions, which can be viewed as small-scale `collusion' between the bids by the same author. 

After collecting bids, the VCG auction allocates the $P$ review slots to the $P$ papers with the highest bids, denoted $ \{\omega^*_1, \ldots , \omega^*_P\}$, with associated bids $\{b_p^*\}_{p\in [P]}$. Denoting the paper with the $P+1$-highest bid ${b}^*_{{P+1}}$ as ${\omega}^*_{P+1}$, the mechanism charges all authors of the papers with winning bids a price of $c_i = \sum_{\ell = 1, \ell \neq i}^{P+1} {b}^*_\ell - \sum_{\ell=1, \ell\neq i}^{P} {b}^*_\ell = {b}^*_{P+1}$. Thus, all papers accepted for review are charged the highest losing bid\footnote{This is correct when viewing papers-as-bidders, but a potential concern when viewing authors-as-bidders is that authors who submit multiple papers may try to bid strategically. Strictly speaking, VCG would require charging author $k$ who gets $p_k$ slots the sum of the $p_k$ highest losing bids excluding their own. In practice, this is unlikely to matter since the number of review slots vastly outnumbers an author's submissions.}, and accepting these papers maximizes author welfare for a given number of review slots $P$. 

Truthful bidding is a \emph{weakly dominant strategy} for VCG, and we expect authors' bids to equal to their value for having their paper reviewed. The natural assumption of identical review slots and prohibition of package bids precludes the need to solve a computationally challenging knapsack problem to identify winners. %
Thus, the \emph{Stage I strategy profile} that maximizes aggregate author welfare under a given number of review slots $P$ is ${\bf s}^{(1)} = \left(\left(\{v_k\eta_i \}_{i \in \{\omega_{i_s}^{k}\}}\right)_{k = 1,\ldots,K}\right)$. The mechanism allocates review slots to the highest bidding papers, which will generally be papers that are likely to be accepted for publication and/or have authors who generate a lot of value as reviewers. Thus, the papers accepted for review should also ideally be of higher quality if the review process screens well for quality. This should also eliminate the need for desk rejection.

\subsection{Stage II: Peer Review}\label{stage2}
We now present the second stage of the mechanism which focuses on incentivizing honest and effortful reviews of the $P$ papers accepted by the VCG auction. The context here is as follows: the bid requirement for the submission process in Stage I auction creates demand for credits, and since this can only be earned by reviewing, it incentivizes researchers to join the reviewer pool. While this is desirable, we must structure the payment for reviews carefully. A flat fee to reviewers could result in reviewers putting in minimal effort and getting paid anyway; if credits could be earned with little effort, bids could be made arbitrarily high, defeating the point of the mechanism. Instead, we need to structure the payment such that honest and effortful reviews are paid more than low-effort or false reviews. This is a challenging problem in the absence of ground truth verification; and our proposed H-DIPP mechanism is a novel solution to this problem, drawing on insights from the Target Differential Peer Prediction (TDPP) \citep{schoenebeck2020two} and Hierarchical Mutual Information Paradigm (HMIP) \citep{kong2018eliciting} (see Appendix \ref{relsec}). The main technical challenge here involves determining what information to reveal during the prediction game, in what order, and how to structure the payments. To our knowledge, our design is the only approach that provides the desired equilibrium results. Note that the H-DIPP mechanism stage describes only the reviewing process for a single paper; we are agnostic about the paper-reviewer matching process, and meta-reviewers can make final accept/reject decisions as usual. We present some fundamentals of the Mutual Information Paradigm \citep{kong2019information} used to develop this theory in Appendix \ref{mipsec}. Proofs are presented in Appendix \ref{proofs}. We now give a more detailed discussion of our reviewing mechanism H-DIPP (Mechanism \ref{algo:duplicate2}).

\subsubsection{Model Preliminaries}

 %We begin by assuming that a given paper $\omega_i$ will be scored on a  by exactly 3 reviewers\footnote{This can be extended to any number of reviewers by cyclically assigning the roles of expert, target, and source.} $A, B, C$. 
Each reviewer rates the paper $\omega_i$ on a common set of $T$ criteria. Each criterion $t$ is rated on a $D_t$-point scale, i.e., reviewers' reported scores can take on any value in a $\Sigma_t = \{1, 2, \ldots, D_t\} $. The same H-DIPP mechanism will be used determine payments for every paper's reviewers, so we drop the index $i$ referring to a specific paper going forward. Now, suppose we have a sample space $\Omega$ of all the possible realizations of papers (`world states') a reviewer may see. A reviewer $k$'s \emph{private review score} on criterion $t$ is the random variable ${X}^t_k: \Omega \rightarrow \Sigma_t$ describing their honest private assessment (or `signal') of the paper on that criterion, \emph{when they exert the necessary effort to arrive at an informed opinion}. A reviewer's joint distribution of such `effort-informed review scores' for all criteria is represented by the multivariate random variable ${\bf X}_k^{[T]} = (X_k^1, X_k^2, \ldots, X_k^T) \in \Sigma_1 \times \Sigma_2 \times \ldots \times \Sigma_T$, and the full joint distribution of effort-informed review scores over all criteria for all reviewers is $\mathds{P}({\bf X}_A^{[T]}, {\bf X}_B^{[T]}, {\bf X}_{C}^{[T]})$. We assume that this joint distribution is common knowledge and  has full support, i.e., every assignment of outcomes has non-zero probability. This is a standard basic assumption in our setting, and necessary in order to make any claims about agents' beliefs about the papers they're reviewing. We discuss how this assumption can be made more realistic in Section \ref{disc}, and other assumptions in  Appendix \ref{otherass}.  %Unlike HMIP, differential peer prediction (and our approach) \emph{does not} require the assumption that the common prior be symmetric, and is applicable for heterogeneous agents.
We state this assumption formally below:
  
   \begin{assumption}[Common Prior with Full Support] \label{comprior}
  Reviewers' effort-informed signals on all criteria $\{{\textbf x}_A^{[T]}, {\textbf x}_B^{[T]}, {\textbf x}_C^{[T]}\}$ are drawn from the common prior $\mathds{P}({\bf X}_A^{[T]}, {\bf X}_B^{[T]}, {\bf X}_{C}^{[T]})$, and the common prior is common knowledge. The common prior also has full support.
 \end{assumption}
 
  \paragraph{Criteria Hierarchy} Next, we construct a common \emph{hierarchy of criteria} that each reviewer will rate the paper on. The idea is structure the hierarchy so %We seek to construct a universal hierarchy for all reviewers by imposing an ordering 
  such that reviewers require less effort to come to an informed private assessment on `lower' level criteria than `higher' level criteria, %Naturally, when a reviewer exerts enough effort to come to an informed private assessment on some higher level criterion, they will also have arrived at an informed private assessment on all lower level criteria, so 
  i.e., the total effort required to arrive at informed private assessments increases up the hierarchy. Thus, we say $t_i > t_j$ if scoring the paper on criterion $t_i$ requires strictly more effort than providing a score on criterion $t_j$. When a reviewer has exerted enough effort to arrive at an informed review score  $x_k^{t}$ on criterion $t$, we say that reviewer has \emph{completed} criteria $[t_k] = \{1, \ldots, t_k\}$. The idea is to explicitly construct lower level criteria to collect the cheap signals and pay agents more for the information \emph{added} by valuable higher level criteria. We can define the positive, {increasing} function $e_k: [T] \rightarrow \mathds{R}^+$ specifying reviewer $k$'s \emph{total effort cost} (in the mechanism-specific credits) to obtain effort-informed signals up to criterion $t$, where we assume $e_k(t) > e_k({t'})$ if $t > t'$. % While we aim to construct a single universal criteria hierarchy where the effort cost is monotonically increasing for all reviewers, the cost function $e_k$ itself may be different for different reviewers; for instance, veteran peer reviewers may require less effort to score papers on every criterion. We give a sample plausible criteria hierarchy below.

 In the current peer review setting at machine learning conferences, reviewers are already required to report scores on a variety of criteria, so we simply require a particular structuring of existing criteria. As an example, one possible hierarchy of criteria could be the following (where responses fall on a 5-point scale): 1) \emph{(presentation)} ``How clear is the writing/exposition?''; 2) \emph{(completeness)} ``How easy would it be to reproduce major results?''; 3) \emph{(correctness)} ``How accurate are the technical claims and methodology?''; 4) \emph{(contribution)} ``How interesting and valuable of a contribution is the work? ''; 5) \emph{(overall quality)} ``In which quintile of submitted papers would this work fall?"

 \begin{figure}[h]
    \centering
    \includegraphics[width=\textwidth, height=55mm]{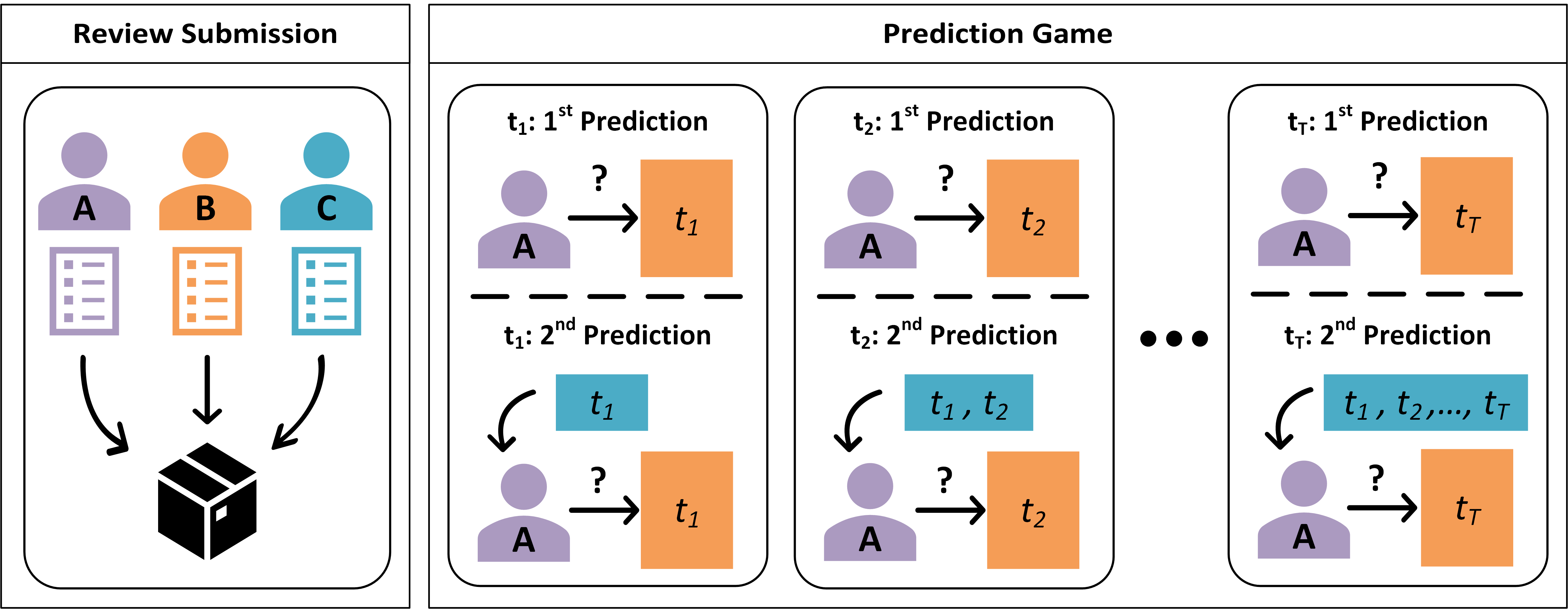}
    \caption{An illustration of the H-DIPP mechanism from reviewer $A$'s perspective, when reviewer $B$ is their assigned target and reviewer $C$ is their assigned source.}
    \label{fig:ill}
\end{figure}

  \subsubsection{Mechanism and Payments} 
 Having detailed our assumptions, we now describe exactly how the mechanism collects reports and computes the payment. The mechanism is summarized as Mechanism \ref{algo:duplicate2} and illustrated in Figure \ref{fig:ill}. We give a step-by-step example in Appendix \ref{app:add}.
 
 \paragraph{Reports} First, each reviewer reports their review scores to the mechanism per the usual process. Reviewers arrive at private review scores ${\bf x}^{[t_k]}_k$ on $t_k$ criteria by exerting $e_k(t_k)$ effort, and report this to the mechanism (potentially strategically) as $\hat{\bf x}^{[T]}_k \sim  \theta_k(\cdot | {\bf x}^{[t_k]}_k)$. 
 
 Once all the reviews are in, every reviewer is matched with two others and (independently) plays a simple prediction game. % We refer to the reviewer playing the prediction game as an `expert', and every `expert' will be assigned a `source' and a `target'. We use the following assignments of (expert, target, source) roles: $(A, B, C)$, $(B, C, A)$, and $(C, A, B)$, i.e., when $A$ is plays the game as `expert', $B$ is the assigned `target' and $C$ is the assigned `source', and so on. During the prediction game, the \emph{expert} must make two probabilistic predictions of their \emph{target's} reported review score on each criterion; the expert's first prediction on a criterion is their best guess of their target's review score on that criterion, and the expert's second prediction is their updated best guess after the mechanism reveals their \emph{source's} review score on that criterion. The criteria are revealed in a pre-determined order specified by the criteria hierarchy.
 Suppose reviewer $A$ plays the game and is matched with reviewers $B, C$. %As an example, consider the case where reviewer $A$ plays the prediction game after all the reviews are in (with target $B$ and source $C$).
 The mechanism will first elicit the reviewer $A$'s prediction $\hat{\bf p}_{B^1 \leftarrow A} = P\left(\hat{X}_B^1|{\bf x}_A^{[t_A]}\right)$ of the reviewer $B$'s score on criterion $1$. Then, it reveals reviewer $C$'s score $\hat{x}_C^1$ on that criterion and asks for an updated prediction from reviewer $A$: $\hat{\bf p}^+_{B^1 \leftarrow A} = P\left(\hat{X}_B^1|{\bf x}_A^{[t_A]}, \hat{\bf x}_C^{1}\right)$. The mechanism repeats this for criterion $2$, $3$, $\ldots$, and so on % on criterion 2, the mechanism asks for the expert $A$'s prediction $\hat{\bf p}_{B^2 \leftarrow A} = P\left(\hat{X}_B^2|{\bf x}_A^{[t_A]}, {\bf x}_C^{1}\right)$ of the target $B$'s score on  criterion $2$, and then reveals the source $C$'s signal $\hat{x}_C^2$ on that criterion and asks for an updated prediction $\hat{\bf p}^+_{B^2 \leftarrow A} = P\left(\hat{X}_B^2|{\bf x}_A^{[t_A]}, \hat{\bf x}_C^{[2]}\right)$. 
 until the $T$th criterion. The mechanism thus requests $3T$ reports: reviewers' $T$ review scores, their first predictions of their assigned target's scores on $T$ criteria, and updated second predictions of that target peer's scores on $T$ criteria given the source peer's scores. A reviewer's strategy on the reviewing task is $s_k = \left(e_k(t_k),\hat{\bf x}_k^{[T]}, \hat{\bf p}_{k'^t \leftarrow k},\hat{\bf p}^+_{k'^t \leftarrow k}\right)$.

\begin{algorithm}[h]
\SetAlgorithmName{Mechanism}{}
\DontPrintSemicolon
\KwIn{$\omega_i$ (paper under review), Reviewer $A$ matched with reviewers $B$ and $C$, $\{t_1, t_2, \ldots, t_T \}$ (hierarchical criteria),  $\{  \boldsymbol{\alpha}_A^{[T]}, \boldsymbol{\beta}_A^{[T]}\}$ (mechanism hyperparameters)}
\KwOut{$r_A$ (payments to reviewer $A$) } 

Elicit reviewer $A$'s scores on each criterion $\{\hat{\bf x}_A^{[T]}\}$ per the usual peer review process \\
% Assign (expert, target, source) roles as follows: $(A, B, C)$, $(B, C, A)$, $(C, A, B)$.\\
\For{each criterion $t$ ascending the criteria hierarchy}{

Elicit reviewer $A$'s predictions of reviewer $B$'s score on criterion $t$ ($\hat{\bf p}_{B^t\leftarrow A}$). \\

Reveal reviewer $C$'s score  $\hat{x}_C^{t}$ on criterion $t$ to reviewer $A$. \\

  Elicit reviewer $A$'s \emph{updated} prediction of reviewer $B$'s score on criterion $t$ ($\hat{\bf p}^+_{B^t\leftarrow A}$).\\
  
  $r_{A^t}^\text{expert} \gets {S_\text{log}(\hat{x}_B^t, \hat{\mathbf p}_{B^t \leftarrow A}) + S_\text{log}(\hat{x}_B^t, \hat{\mathbf p}^+_{B^t \leftarrow A})} $ \\
  
 $ r_{A^t}^\text{target} \gets S_\text{log}(\hat{x}_A^t, \hat{\mathbf p}_{A^t \leftarrow C}^+) - S_\text{log}(\hat{x}_A^t, \hat{\mathbf p}_{A^t \leftarrow C}) $
%   $r_{B^t}^\text{expert} \gets {S_\text{log}(\hat{x}_C^t, \hat{\bf p}_{C^t \leftarrow B}) + S_\text{log}(\hat{x}_C^t, \hat{\bf p}^+_{C^t \leftarrow B})}; r_{B^t}^\text{target} \gets S_\text{log}(\hat{x}_B^t, \hat{\bf p}_{B^t \leftarrow A}^+) - S_\text{log}(\hat{x}_B^t, \hat{\bf p}_{B^t \leftarrow A}) $ \\ 
%   $r_{C^t}^\text{expert} \gets {S_\text{log}(\hat{x}_A^t, \hat{\bf p}_{A^t \leftarrow C}) + S_\text{log}(\hat{x}_A^t, \hat{\bf p}^+_{A^t \leftarrow C})}; r_{C^t}^\text{target} \gets S_\text{log}(\hat{x}_C^t, \hat{\bf p}_{C^t \leftarrow B}^+) - S_\text{log}(\hat{x}_C^t, \hat{\bf p}_{C^t \leftarrow B}) $  
%  
}
$ r_A \gets \sum_{t=1}^T \alpha_A^t r_{A^t}^\text{expert} + \beta_A^t r_{A^t}^\text{target} $
% $ r_B \gets \sum_{t=1}^T \alpha_B^t r_{B^t}^\text{expert} + \beta_B^t r_{B^t}^\text{target} $ \\
%  $r_C \gets \sum_{t=1}^T \alpha_C^t r_{C^t}^\text{expert} + \beta_C^t r_{C^t}^\text{target} $
\caption{H-DIPP}
\label{algo:duplicate2}
\end{algorithm}

 \paragraph{Payments} The mechanism will pay each reviewer for their roles as expert and target (source payment is 0). Reviewers receive `expert payment' based on their performance in the prediction game, and receive `target payment' based on how reliably their scores relate to a peer's \citep{schoenebeck2020two}. 
 With mechanism hyperparameters $\{\boldsymbol{\alpha}_A^{[T]}, \boldsymbol{\beta}_A^{[T]}\}$, the criterion-$t$ specific payment to reviewer $A$ using the log scoring rule $S_\text{log}$ is:  
 
\begin{equation} \label{rew} \small
    r_{A^t} = \alpha_A^t\left(\underbrace{S_\text{log}(\hat{x}_B^t, \hat{\bf p}_{B^t \leftarrow A}) + S_\text{log}(\hat{x}_B^t, \hat{\bf p}^+_{B^t \leftarrow A})}_{\text{expert payment}}\right) + \beta_A^t\left(\underbrace{S_\text{log}(\hat{x}_A^t, \hat{\bf p}_{A^t \leftarrow C}^+) - S_\text{log}(\hat{x}_A^t, \hat{\bf p}_{A^t \leftarrow C})}_{\text{target payment}} \right)
\end{equation}

A reviewer's \emph{expert payments} depend on how good their predictions were during the prediction game (incentivizing truthful predictions). A reviewer's \emph{target payments} depend how much better a peer expert's second prediction of their reported review score was over their first;  the intuition is that their peer's second prediction can only update reliably towards the target reviewer's true review score (thus incentivizing truthful reviews), given a third reviewer's score (see Proof of Theorem 1 in Appendix \ref{proofs}).  The mechanism hyper-parameters are used to appropriately weight the criterion-specific payments to cover the marginal effort costs of completing higher level criteria (see Theorem \ref{optparam} in Appendix \ref{app:add}). To compute reviewer $A$'s total payment, we simply sum the individual criterion-specific payments $r_A = \sum_t
 r_{A^t}$. % Payments for the other reviewers $B, C$ are determined analogously.  %The key idea behind the target payment is that in expectation (shown shortly in Lemma \ref{revwelf}), it pays a reviewer the mutual information of their private score and the source's score, conditioned on the expert's information and the source's scores on lower level criteria. it is the \emph{additional} information the second prediction contains about the source's score on that criterion, over and above the source's cheaper, lower effort information and expert's information on the lower level criteria.
 
 % The mechanism hyper-parameters are used to appropriately weight the criterion-specific payments to cover the marginal effort costs of completing higher level criteria, and will describe how to set them in relation to reviewers' common prior and effort costs in the next subsection. The task of gleaning information about the common prior and reviewers' effort costs is important, but we leave this to future work.

 \subsubsection{Equilibria}
 
We now analyze the equilibria of our mechanism\footnote{We ignore permutation equilibria, where the labels for signals are permuted. Such equilibria exist in our mechanism, but they are not of serious concern; see \citet{kong2019information} for details.}. First, we give the strict truthfulness of our mechanism when all reviewers complete a certain number of criteria.
 
\begin{theorem}[Strict Truthfulness]\label{truththm}
Suppose reviewers $A, B, C$ complete $[t_A], [t_B]$, and $[t_C]$ criteria respectively. Then, on shared complete criteria $[t'] = [\min(t_A, t_B, t_C)]$, the H-DIPP mechanism is strictly truthful, i.e., all reviewers honestly reporting their true predictions and private review scores on shared complete criteria $[t']$ is a strict Bayesian Nash equilibrium.
 \end{theorem}
 
However, mere truthfulness is not enough; the mechanism must also incentivize reviewers to exert the necessary effort to complete all criteria. After all, we care about eliciting reviewers' \emph{effort-informed signals} on all criteria. To achieve this, we reason about reviewers' marginal payments and costs of effort to determine mechanism hyperparameters in Appendix \ref{app:add}. % After all, we care about eliciting reviewers' \emph{effort-informed signals} on all criteria. To achieve this, we must reason about reviewers' marginal payments and costs of effort. Going forward, we use $R$ to denote expected payment. We break up the total marginal expected payment $\Delta R_k$ into the marginal expected expert and target payments $\Delta R_k^\text{expert}$ and $\Delta R_k^\text{target}$, which in turn are each broken up into marginal payments $\Delta R_{k^t}^\text{expert}$ and $\Delta R_{k^t}^\text{target}$ on each criterion. We give the following lemma on reviewers' marginal expected payments for completing an additional criterion and show that it is always non-negative (when peers have completed all criteria and are reporting truthfully). 

 Note that a reviewer's expected expert payment is always non-positive (since the Shannon entropy is non-negative), and a reviewer's expected target payment is always non-negative. We must ensure that a reviewer's expected payment is greater than their cost in the fully informative equilibrium, so that it is individually rational for reviewers to participate in the mechanism. We can achieve this by  scaling the hyperparameters such that the target payment pays more than the expert payment subtracts. This also grants that the fully-informative equilibrium pays more than an uninformative equilibrium where no reviewers exert effort. We state this formally below.
 
 \begin{theorem}[Uninformative Equilibria and Individual Rationality]\label{ir} \phantom{=}
 
 \begin{enumerate}[label=(\alph*)]
     \item In the highest paying uninformative equilibrium (where all reviewers exert zero effort), every reviewer's expected payment (and hence aggregate reviewer welfare) is zero.

     \item % Suppose that in addition to the constraints specified in Theorem \ref{optparam}, the mechanism hyperparameters also satisfy $R_k(T) - e_k(T) > 0 \Leftrightarrow \sum_{t=1}^T\alpha_k^t R_{k^t}^\text{expert}(T) + \beta_k^t R_{k^t}^\text{target}(T)  > e_k(T)$ for $k \in \{A, B, C\}$, i.e., 
     Suppose we set mechanism hyperparameters so reviewers' expected utility in the fully-informative equilibrium is positive (see Theorem \ref{optparam} in Appendix \ref{app:add}). Then, it is individually rational to participate in the mechanism and complete all criteria and report truthfully, given peers are doing the same. Additionally, the fully-informative equilibrium has higher individual and aggregate utility than the uninformed equilibrium.
 \end{enumerate}
 \end{theorem}

% Finally, we note that while the mechanism can hide reviewers' identities from each other to provide some measure of protection against collusion, a formal analysis of H-DIPP's robustness to collusion would be valuable direction for future work.

\section{Related Work}

There is an immense amount of scholarship on improving the peer review process. Here, we discuss three bodies of related work: 1) proposals to improve scientific peer review from computer science, 2) general proposals to reform or replace scientific peer review, and 3) the information elicitation without verification literature \citep{faltings2017game, waggoner2013information}. We refer readers to \citet{shah2022challenges} for a comprehensive overview of the challenges facing peer review today. Readers may also be interested in a broader historical overview of the development of peer review (focused primarily on the social sciences) presented by \citet{merriman2020peer}.

\subsection{Proposals from Computer Science}

A number of proposals to handle the rapid rise in submissions and reviewing load have been put forward by the machine learning community. One idea is to allow some editorial screening, where Area Chairs can reject submissions without full reviews; this would eliminate the need for $\sim 4\%$ of reviews but some papers may still be screened out unfairly. Another proposal is to introduce submission caps, which limit the number of submissions by a single author, but a cap of 10 submissions would again only eliminate $\sim 4\%$ of the reviewing load. An analysis of NeurIPS 2019 data \citep{neural} concluded that it was unclear how best to rapidly filter papers prior to full review; our proposed submission mechanism uses a VCG auction to determine which papers are accepted for review, where authors signal how ready their work is and the mechanism accepts as many papers for review as it can find quality reviewers for. 

Criticizing the slow pace of peer review, \citet{lecun} and \citet{Zhao2012TwoPF} propose an `open-market', where all papers are uploaded to an online repository as soon as they are ready for viewing.  `Reviewing entities' may then identify papers they wish to review, although authors may request reviews as well. However, data from past NeurIPS conferences \citep{shah2018design, neural} suggest that under the current implementation, reviewers do not make many eager bids to review papers, and eager bids are not predictive of eventual acceptance. Thus, mechanisms that rely on reviewers to identify and review papers may need to provide explicit incentives. Identifying the tension between rapid dissemination of ideas and the need to carefully screen results, \citet{bengio_2020} proposes a hybrid arXiv-journal model, where arXiv addresses the former need and journals take on the role of slow, careful screening of results to promote higher-quality publications. In these aforementioned approaches, reviewing incentive issues would still persist (as noted in fields with journal publication models); indeed, the reviewing burden may move from conferences to the open-market system or journals, as authors from less well-known backgrounds seeking community-issued certificates of the quality of their work for reputation-building purposes  begin to substitute conference submissions for journal submissions. Our submission mechanism disincentivizes weak submissions that add to the reviewing load while being unlikely to be accepted, and our reviewing mechanism provides direct incentives for honest and effortful reviews. 

Another particular area of focus has been on improving the reviewer-paper matching process \citep{stelmakh2023gold}. Several recent works in this space include a market-inspired approach assigning budgets and prices to incentivize reviewers to bid for their favourite papers \citep{meir2021market}, optimizing reviewer selection or assignment to obtain less noisy reviews \citep{shah2019principled} or better matching when the assignment is over multiple rounds \citep{jecmen2021near}, mitigating conflicts of interest \citep{kurokawa2015impartial, xu2018strategyproof,  aziz2019strategyproof, jecmen2020mitigating}, and determining the optimal way to display papers to reviewers to elicit meaningful bids \citep{fiez2020super}.  Other proposals involve improvements to various procedural aspects of peer review, such as a recruitment and mentorship pipeline for junior academics to alleviate the scarcity of reviewers \citep{stelmakh2020novice} and better aggregation of subjective opinions on criteria scores to accept/reject recommendations \citep{noothigattu2021loss}.  Such ideas are complementary to our reviewing mechanism; they focus on improving paper-reviewer matching or the acceptance decision process, while we focus on incentivizing honest and effortful reviews after reviewers are assigned and before scores are aggregated into accept/reject decisions.
    
Lastly, we mention relevant empirical evidence from the NeurIPS 2021 review process \citep{rastogi2022authors} which finds that authors systematically overestimate their paper's probability of acceptance, and suggests genuine uncalibration, non-response bias, and a potential desire to signal confidence to reviewers as possible explanations. Verifying if these results hold even when authors have explicit incentives to provide accurate forecasts would inform the viability of our proposed submission stage.

\subsection{General Proposals to Reform or Replace Peer Review}

Most relevant to this paper is the empirical work by \citet{chetty2014policies} studying interventions to promote pro-social behavior in peer review setting. They found that shorter deadlines, cash incentives, and social incentives all effectively reduced review times, and notably, cash incentives did not undermine intrinsic motivation to review. These interventions incurred minimal costs and did not impact review acceptance rates, report quality, or review times at other journals.

Another strain of research proposes modest changes to the system, such as allowing reviewers to re-purpose reviews from past submissions \citep{alberts2008reviewing} to reduce the demand for reviews, or improving the publication culture so authors submit more complete and relevant work to the appropriate venues instead of simply targeting the most prestigious venues.  Others \citep{arns2014open} propose reducing the reviewing burden by allowing some papers (e.g. null results) to be accepted without peer review. However, these proposals do not systematically tackle the previously discussed core incentive issues.

\citet{fox2010pubcreds} propose an approach similar to ours, where authors must pay in an artificial currency to have their paper reviewed, and reviewers are paid for their completed reviews. However, reviewers receive flat payment for their participation, and so may not provide honest or effortful reviews. \citet{sculley2018avoiding} propose pecuniary and non-pecuniary compensation for reviewers based on review quality (as judged by a rubric), but rating reviews would increase the workload on area chairs.

 In a different approach, \citet{de2007selecting, de2007citation} propose entirely \emph{replacing the existing peer review system} with market-based citation auctions, i.e., where earned citations are treated as currency, every submission comes with a ``citation bid", and the papers with the highest bids are accepted for publication. The citations that the paper goes on to earn are used to recoup the bet, with the idea being that the bid will reveal how impactful (as measured by number of citations) a researcher expects their work to be and the mechanism will publish papers expected to be highly cited. %They attempt to use citation count to avoid the fundamental problem of verifying review quality, and citations certainly provide some crowd-sourced opinion on the quality of a paper. 
 However, even as this fixes the shortage of reviewers by getting rid of the need for them in the first place, predicting citations is likely extremely high variance, can lead to distorted citation practices, and poses equity issues if highly-cited authors can publish any paper they want by simply outbidding everyone else regardless of the quality of their work.

 \citet{prufer2010auction} also propose a citation-based paper submission system for journals, although their proposal does not abolish academic peer review. Instead, they propose a new currency of `academic dollars' used by \emph{editors} of various journals to bid for submitted papers in an auction; if accepted for review, the dollar value is paid to the authors, referees, and editors of \emph{cited} works. This is essentially a citation betting market for journal editors; a paper's bid will be commensurate with its future expected citation count. Referees will exert effort to provide better feedback to authors to improve the quality of their work in order to earn dollars from future citations, and authors have a marginally higher incentive to write better to earn citations. While an intriguing proposal, we suspect the task of predicting citations is challenging and the singular focus on optimizing for citations may distort citation practices.

 \citet{frijters2019improving} also propose a modified peer review incorporating market-based ideas, with a focus on mitigating `gate-keeping' by insiders and improving transparency. Their proposal includes allowing authors to bid for reviews of varying quality, and opening up the reviews themselves to be rated thereby allowing for `professionalization' of the process. However, the system is somewhat complicated and the market for rating reviews may not be very liquid, especially in technical fields. \citet{xiao2018incentive} propose a peer review mechanism to elicit effortful reviews by assigning papers to reviewers based on the quality of the authors' past reviews, but it computes this review quality by asking authors to rate the reviews they received and assuming these ratings to be honest. This is likely too strong an assumption in our setting where authors may simply report how favourable a review was; this would incentivize flattering and likely dishonest reviews, undermining the mechanism.

\subsection{Information Elicitation without Verification}

The IEWV literature is focused on the problem of eliciting truthful private `signals' when the mechanism designer is unable to verify the responses; this may be because the ground truth is difficult or impossible for the mechanism to access or there is simply no ground truth since the responses are subjective. The key application domains of interest include tasks like crowdsourcing data, product ratings, community sensing, and peer grading \citep{cai2015optimum, radanovic2016incentives, dasgupta2013crowdsourced, kong2018eliciting}. More recent work is also interested in ensuring that private signals are obtained with effort.  The typical setting in this literature consists of risk-neutral, utility-maximizing agents who maintain some prior distribution over `signals' they can receive from a mechanism. When the agents participate in the mechanism, the world reveals to them their private information or `signal', and they subsequently form posterior beliefs over the distribution of signals observed by the other agents. In our setting, this would be analogous to reviewers having some prior belief over the quality of papers they may see, and when presented with a specific paper, the reviewer forms their opinion of that particular paper and develops a posterior belief over other reviewers' assessments of the paper. The task of the mechanism designer is to elicit agents' private signals when verification is not possible by providing incentives for truthful reports.

 Two of the foundational approaches to this problem are the peer prediction framework \citep{miller2005eliciting} and the Bayesian Truth Serum (BTS) \citep{prelec2004bayesian}. The fundamental idea is to design payments by reasoning about agents' reports and posteriors over other agents' reports in a way that makes truth-telling a Bayesian Nash equilibrium. 
 However, these mechanisms make strong assumptions that are unlikely to hold in practice; subsequent works  have relaxed these assumptions  \citep{witkowski2012peer, witkowski2013learning, witkowski2012robust, radanovic2013robust, radanovic2014incentives2}. There are also \emph{output agreement} mechanisms \citep{waggoner2014output} with weaker assumptions, but these eschew truthfulness and seek to elicit common knowledge. 
 A directly relevant proposal to our work is by \citet{carvalho2013inducing}, which proposes a peer prediction-like mechanism for the same task of academic peer review. They assume that reviewers' signals on a given paper are drawn from a multinomial distribution, and model reviewers' prior over the multinomials with the conjugate (Dirichlet) prior so inference is tractable, but this a somewhat restrictive assumption. 
 They further test their mechanism on an experiment with Amazon Mechanical Turk workers, and find that their mechanism does incentivize more accurate reviews and such reviews also produce higher payments for the participating agents.  The truthful mechanism applicable in the most general setting is the recently proposed prediction market-style Differential Peer Prediction (DPP) mechanism \citep{schoenebeck2020two}. The truth-telling equilibrium maximizes agent welfare, and the mechanism is applicable to a small group of heterogeneous agents. Our proposed reviewing mechanism is a direct extension of this work.

 These aforementioned works do not incorporate any model of agent effort, and instead assume that nature or some black box simply reveals the private signal to the agents. The early literature incorporating effort models \citep{dasgupta2013crowdsourced,  witkowski2014dwelling, radanovic2016incentives, liu2016learning} is limited to the case of binary effort: either agents exert low/no effort, or they exert high/full effort.  
A binary effort model may be appropriate for quick tasks like those on crowdsourcing platforms, but reviewing technical papers is a time-intensive process, so a more sophisticated effort model is needed. \citet{cai2015optimum} propose a continuous effort model for data elicitation, where exerting effort gives agents a lower variance signal. However, in our setting, lower effort reviews may not only be noisier, they may also be systematically biased \citep{kong2018eliciting}.  
\citet{gao2016incentivizing} caution that in practice (and as highly relevant in our case), agents have a variety of low-effort signals such as paper length, topic, etc., potentially leading to low-effort coordination into uninformative equilibria. Indeed, \citet{gao2014trick} present empirical information showing that agents can successfully coordinate into uninformative equilibria, especially if these equilibria pay more than truthful equilibria. 

With these considerations, \citet{kong2018eliciting} propose a Hierarchical Mutual Information Paradigm (HMIP) that constructs a hierarchy of criteria that capture a range of cheap to expensive signals, and pays agents the \emph{information gained} from more expensive signals. This work builds on the Mutual Information Paradigm 
\citep{kong2019information} which shows that mechanisms incentivize truth-telling when the payment is based on the mutual information between agents' responses, and this is what peer prediction and BTS are attempting to do. We integrate the very general differential peer prediction (DPP) mechanism into the hierarchical mutual information paradigm (HMIP) to develop a strictly truthful prediction market-style mechanism that incentivizes effortful reviews, with the desired equilibria paying more than uninformative equilibria.

\section{Discussion}\label{disc}

We tackle incentive issues in scientific peer review by proposing a two-stage mechanism that allocates review slots to papers via a VCG auction and uses the revenue raised to pay participating reviewers  in our novel H-DIPP mechanism.  The H-DIPP mechanism aims to incentivize reviewers to obtain and report a clearer signal of a paper's quality in order to reduce the noise in the accept/reject decisions. We showed that the H-DIPP mechanism is strictly truthful, and payoffs can be designed so that exerting enough effort to score the paper honestly on all criteria is a strict Bayesian Nash equilibrium. We also showed that such a fully informative equilibrium pays more in expectation than any uninformative equilibrium, and described how the mechanism hyperparameters can be further tweaked so participating in the H-DIPP reviewing mechanism is individually rational. %Finally, we discussed the trade-off between the number of available review slots and quality of paper-reviewer matches. 
However, there are many important considerations for a practical implementation of our mechanism, some of which we discuss below (see Appendix \ref{conc} for concrete suggestions for future work). 

\paragraph{Fairness and Access Considerations} Our mechanism requires that each submitted paper is accompanied by a bid from a single author. %assumed for ease of exposition that every paper has a single author and this author must participate in the review process to earn enough credits to cover their bids. 
In practice, papers are largely multi-author efforts, so the question of who pays the bid naturally arises. Not every co-author is equally qualified to participate in the review process; research and review may require different skills. Furthermore, new or under-resourced authors may not be able to afford strong bids%The natural relaxation of the single-author assumption is to formalize the role of a `sponsor' 
\footnote{A paper should only have a single bid sponsor; allowing multiple sponsors to contribute to a bid could unfairly privilege papers with more authors or those from larger research groups. For the same reason, credits should not be freely giftable to others.} as they may not have been able to build up a store of credits from reviewing. Thus, we will need a few interventions to ensure the proposed mechanism does not become an undue barrier to publishing academic research. One suggestion is to set a norm that the bid sponsor typically be a senior author who has banked enough credits by reviewing. We may also wish to allow non-author researchers to sponsor the credits for a paper's accompanying bid and be acknowledged in their capacity as sponsor. %, mitigating incentives to add credit-rich non-author sponsors as authors in exchange for a high bid.
Additionally, we suggest granting a lump-sum of credits to first-time or junior authors to give them the right to submit some number of papers in a given time period while they build the necessary expertise for review. % Ultimately however, we are trying to solve the mismatch between submissions and reviewers, so we want an incentive for authors to learn to review. % We emphasize that bids do not have any bearing on a paper's acceptance for publication, so sponsors will still be sensitive about their bids as a weak paper accepted for review will likely be rejected and the spent credits will be wasted. 

\paragraph{Common Prior Assumption} We assumed that reviewers' review scores are drawn from a (potentially asymmetric) common prior that is common knowledge, a standard assumption in the IEWV literature. This assumption is important to obtain the equilibrium results, but the extent to which it holds in practice needs to be answered empirically. This assumption can perhaps be made more realistic by making examples or summary statistics of a reviewer's past reviews available for reference (so any reviewer can gauge how their opinion correlates with their peers’) .

\paragraph{Post-processing Review Scores} The H-DIPP mechanism elicits honest and effortful opinions from \emph{heterogeneous} reviewers, with the goal of reducing the noise in estimates of a paper's quality. However, this also means that reviewers' honest and effortful opinions may be miscalibrated with respect to the rating scale \citep{shah2018design}. Thus, some post-processing of these scores \citep{wang2018your, lawrence_2015} to correct reviewer biases may produce a better picture of a paper's quality.

\paragraph{Conflicts of Interest} We note that our results on incentivizing effortful and truthful reviews hold under the assumption that there are no other incentives in play. However, if the matched reviewers have a conflict of interest, this would weaken the incentive to report honest and effortful reviews. Thus, related work on optimizing paper-reviewer matches \citep{xu2018strategyproof, jecmen2020mitigating} to avoid conflicts of interest is complementary and important to the success of our mechanism.

\paragraph{ A Secondary Market} We have assumed that our mechanism operates with its own currency, but we could further allow the sale and purchase of credits on a secondary market. This has the potential to allow individuals to earn real-world currency as reviewers, and could even encourage specialization in the review process: researchers who excel at reviewing in H-DIPP mechanism could allocate more of their time providing high quality reviews and be appropriately compensated for their work.

\paragraph{Behavioural and Social Considerations} Our focus in this paper has been to develop a mechanism for paper submission and reviewing with desirable theoretical properties. However, our fundamental assumptions about human behaviour may or may not hold in practice. A primary concern is that introducing payments for reviewers could crowd-out pro-social motivations to review \citep{squazzoni2013does, gneezy2000pay};  \citet{chetty2014policies} empirically study precisely this question and find that ``cash payments do not crowd out intrinsic motivation" and have ``little or no effect on rates of agreement to review, quality of reports, or review times at other journals".  Further empirical evidence for when such peer prediction mechanisms in particular succeed \citep{debmalya2020effectiveness} and fail \citep{gao2014trick} can help refine our mechanism and similar ideas.

Such issues merit further in-depth discussion and analysis. Our proposed mechanism is intended to serve as a theoretically grounded approach to tackling incentive issues in academic peer review, and we expect that a practical implementation of our proposed mechanism will require ad-hoc changes. Ultimately, experimental validation will be the final arbiter of the plausibility of our approach.

\begin{ack}
We thank Sandesh Adhikary for assistance with figures in the paper.
\end{ack}

\newpage

\bibliographystyle{apalike}
\bibliography{biblio}

\newpage

\appendix

\section{Table of Notation}\label{notapp}

\begin{table}[h]
  \label{sample-table}
  \centering
  \begin{tabular}{cl}
    \toprule
    {\bf Symbol}     & {\bf Description}    \\
    \midrule
    $N$ & Number of papers submitted to the mechanism       \\
    $K$ & Number of agents participating in the mechanism       \\
    $T$ & Number of criteria reviewers must evaluate a paper on \\
    $i$            & Index for a paper \\
    $k$ & Index for a reviewer \\
    $t$ & Index for a criterion \\
    $\omega_i$            & The $i$-th paper submitted to the mechanism \\
    % $\omega_{i_s}^k$            & The $i_s$-th paper submitted to the mechanism by reviewer $k$  \\
    % $\omega_{i_r}^k$            & The $i_r$-th paper submitted to the mechanism by reviewer $k$  \\
    $s_k^{(\cdot)}$ & Agent $k$'s Stage I/II strategy \\ % of bids on all submitted papers \\
    % $s_k^{(2)}$ & Agent $k$'s Stage II strategy of scores and predictions on all reviewed papers\\
    $s_k$ & Agent $k$'s full Stage I and Stage II strategy \\
    ${\bf s}$ & Full strategy profile of all agents on Stage I and Stage II \\
    ${\bf s}_{-k}$ & Strategy profile of all agents except agent $k$\\
    $U_k^{(\cdot)}$ & Agent $k$'s Stage I/II utility \\
    $U_k$ & Agent $k$'s total utility \\
    \midrule 
    $b_i$  &  The bid accompanying paper $\omega_i$ for a review slot \\
    $c_i$  &  The price charged for accepting paper $\omega_i$ for review \\
    $b^*_p$  &  The $p$-th highest bid \\
    $\omega^*_p$  &  The paper with the $p$-th highest bid \\
    $\delta_i$  &  Boolean denoting whether paper $\omega_i$ is accepted \emph{for publication} after review\\
    $\eta_i$  & Probability that paper $\omega_i$ is accepted \emph{for publication} conditional on review\\
    $v_k$ & Agent $k$'s value for a paper accepted for publication \\
    $P$ &  The number of papers accepted for review \\
    \midrule 
    $A, B, C$ & The reviewer triplet assigned to a given a paper \\
    $\Sigma_t$ & Set of allowed review scores on criterion $t$ \\
    $e_k(t)$ & Reviewer $k$'s cost for completing up to $t$ criteria \\
    $\mathds{P}$ & A reviewer triplet's common prior over effort-informed scores\\
    $\mathds{Q}$ & A reviewer triplet's joint distribution over reported scores and effort-informed scores \\
    ${\bf X}_k^{[t]}$ & Reviewer $k$'s true review scores on criteria $1, \ldots, t$  \\
    $g_k^{t'|t}({\bf x}_k^{[t]})$ & Reviewer $k$'s `best guess' of their signal on criterion $t'$ given signals ${\bf x}_k^{[t]}$ (where $t' > t$) \\
    $\hat{\bf X}_k^{[T]}$ & Reviewer $k$'s (potentially strategically) reported review scores on criteria $1, \ldots, T$ \\
    $\theta_k$ & Transition matrix representing reviewer $k$'s signal reporting strategy \\
    $\hat{\bf p}_{k'^{[T]} \leftarrow k}$ & Probability distribution of reviewer $k$' first prediction of $k'$'s scores on criteria $1, \ldots, T$ \\
    $\hat{\bf p}^+_{k'^{[T]} \leftarrow k}$ & Probability distribution of reviewer $k$'s second prediction of $k'$'s scores on criteria $1, \ldots, T$ \\
    $r_k$ & Reviewer $k$'s payment for reviewing \\
    $R_k(t)$ & \multirow{2}{28em}{Reviewer $k$'s expected payment for completing $t$ criteria when peers complete $T$ criteria and report truthfully}  \\
    & \\
    $R_{k}^\text{role}(t)$ & \multirow{2}{28em}{Reviewer $k$'s expected expert/target payment for completing $t$ criteria when peers complete $T$ criteria and report truthfully}  \\
    & \\
    $R_{k^{t'}}^\text{role}(t)$ & \multirow{2}{28em}{Reviewer $k$'s expected expert/target payment \emph{on criterion $t'$} for completing $t$ criteria when peers complete $T$ criteria and report truthfully}  \\
    & \\
    $\Delta R_k(t)$ & \multirow{2}{28em}{Reviewer $k$'s marginal expected payment for completing criterion $t$ when peers complete $T$ criteria and report truthfully} \\
     & \\
    $\Delta R_k^\text{role}(t)$ & \multirow{2}{28em}{Reviewer $k$'s marginal expected expert/target payment for completing criterion $t$ when peers complete $T$ criteria and report truthfully} \\
    & \\
    $\Delta R_{k^{t'}}^\text{role}(t)$ & \multirow{2}{30em}{Reviewer $k$'s marginal expected expert/target payment \emph{on criterion $t'$} for completing criterion $t$ when peers complete $T$ criteria and report truthfully} \\
     & \\
    $S_\text{log}(\cdot, \cdot)$ & Log proper scoring rule \\
    $I(\cdot)$ & Mutual information \\
    \bottomrule
  \end{tabular}
\end{table}
\newpage

\section{The 2014 NeurIPS Experiment} \label{2014exp}

In 2014, program chairs of the NeurIPS conference sought to study the consistency of the peer review process \citep{lawrence_2015}, and set up an experiment where 10\% of submitted papers underwent review by two independent committees. We describe the major findings below:
\begin{itemize}
    \item {\bf Inconsistency of Reviews}: About 26\% of papers received different outcomes from the two committees (compared to 37.5\% expected from random acceptance).
    \item {\bf Accept Precision}: Of the accepted papers, only about 50\% of papers would have survived re-review by an independent committee (compared to  25\% expected from random acceptance). 
    \item {\bf Reject Precision}: Of the rejected papers, about 80\% would have also been rejected in a re-review by an independent committee (compared to 75\% expected from random rejection).
 \end{itemize}
 These findings showed that the \emph{accept/reject} decisions were moderately better than pure chance, although still a fair bit noisy. Some of this is unavoidable in any highly selective mechanism seeking to allocate a limited number of spots to a large number of candidates on the basis of some noisily observed measure of `quality' \citep{frank2016success}. When a mechanism seeks to accept some small percentage of top quality candidates, even a small amount of noise in observing the quality of a candidate can add a great deal of randomness to which candidates make the cut. Yet it will be the case that the bulk of the randomness will be faced by the candidates clustered around the cutoff (`messy middle'); candidates well above the cutoff will likely be consistently accepted, and candidates well below the cutoff will likely be consistently rejected. Thus, the noise in observations of paper quality is an important factor affecting the randomness faced by authors in a highly selective acceptance process. \citet{shah2018design} find that about $45 \%$ of NeurIPS 2015 papers and $30\%$ of NeurIPS 2016 papers are in `messy middle' where `inter-reviewer agreements are near random'. Back in the 2014 NeurIPS experiment, an analysis of the independent committees' ratings of paper quality, finds the following:
 
 \begin{itemize}
     \item {\bf Correlation of Quality Scores}: The correlation between calibrated `quality' scores assigned by the two independent committees (adjusted for a reviewers' systematic biases) was 0.546.
     \item {\bf Citation Impact}: For accepted papers, there was little correlation (0.051) between reviewers' assessments of paper quality and papers' log-citation count (likely due to the restricted range effect). For rejected papers, there was a positive correlation (0.22) between reviewers' assessments of paper quality and papers' log citation count. 
 \end{itemize}

 The finding that about half the variance in scores is intrinsic to the paper (`objective') and half is idiosyncratic to the reviewer (`subjective') \citep{lawrence_2015} provides more information on the noise in observations. This subjectivity can arise from more readily identifiable sources like differences in effort or expertise, or from truly subjective differences of perspective. Naturally, if reviewers exert different levels of effort or have differing levels of expertise, their reviews may be systematically different from each other and would weaken the `objective' proportion of review scores. The correlation between quality scores and impact also suggests that the objective proportion of review scores has some better-than-random ability to predict the long-term impact of papers as measured by citations (relevant for criticism C4). In sum, the above results on accept/reject decisions as well as reviewers' quality scores paint a picture of a peer review system that performs noticeably better than random, but with a lot of room for improvement.
 
These results provide a baseline against which to judge the performance of our proposed mechanism. The goal then is for our proposed mechanism to incentivize reviewers to coordinate into a high-effort equilibrium where the `subjectivity' of reviewers' scores is lower, the correlation between independent reviewers' quality scores is higher, and the messy middle is smaller. The hypothesis is that the provision of incentives for effort \emph{during the review process} may help reviewers better observe the quality of papers, and reduce the variance of decisions under the current \emph{selective acceptance process}. 
 
 Finally, we note that conference selectiveness serves the dual purposes of focusing the community's attention on major results (motivation M2), as well as building reputation for authors who do high-quality work (motivation M3). While \citet{lawrence_2015} argues against using conference publications as a proxy for research quality or career advancement, decision-makers looking to make hiring decisions or determine what research to focus on will implicitly or explicitly be seeking relevant signals; in the absence of an alternative signal, we argue that it is important to ensure a peer-reviewed publication gives an appropriate signal of quality.

\section{Key elements of the Mutual Information Paradigm  \citep{kong2019information}}\label{mipsec}

 A scoring rule $S: \Sigma \times \Delta_\Sigma \rightarrow \mathds{R}$ is a function that accepts some outcome $\sigma$ (from some signal space $\Sigma$) and a probability distribution $\delta$ (representing a forecast over that signal space $\Delta_\Sigma$) and scores the forecast with a real number. The scoring rule is \emph{proper} if whenever the outcome $\sigma$ is drawn from a distribution $\delta \in \Delta_{\Sigma}$, the expected score $\mathds{E}_\sigma[S(\sigma, \cdot)]$ is maximized when the forecast is $\delta$. The scoring rule is \emph{strictly proper} if this maximum is unique. Our mechanism will use the log scoring rule $S_{\log}(\sigma, \delta) = \ln \delta(\sigma)$, which is a strictly proper scoring rule. Next, recall that the Shannon Mutual Information of two random variables is defined as the KL divergence of the product of their marginal distributions from the joint distribution: $I(X; Y) = D_{\text{KL}}(P(X, Y) || P(X) \otimes P(Y) )$. The conditional Shannon mutual information is given as $I(X;Y|Z) = \sum_z P(z) D_{\text{KL}}(P(X, Y | z) || P(X | z) \otimes P(Y | z) ) = \sum_z P(z) I(X; Y |z )$. \citet{kong2019information} show an interesting connection between conditional Shannon mutual information and the log scoring rule: 
 
 \begin{equation} \label{logscore}
 I(X; Y |Z) = \mathds{E}_{X, Y, Z}[S_{\text{log}}(Y, P(Y | X, Z)) - S_{\text{log}}(Y, P(Y | Z))]
 \end{equation}
 Finally, we present the important data processing inequality, typically useful in proving that deviating from truth-telling decreases payment. The inequality states that whenever a signal $Y$ is manipulated in a way independent of another signal $X$, the manipulated signal can only have lower mutual information with $X$.

\begin{lemma}[Strict Data Processing Inequality \citep{schoenebeck2020two}] \label{dataproc}

If $(X, Y)$ on a finite space $\mathcal{X} \times \mathcal{Y}$ is stochastic relevant (i.e., for any distinct $x, x' \in \mathcal{X}$, $P(Y|x) \neq P(Y|x')$) and has full support, then for any random function $\theta: \mathcal{Y} \rightarrow \mathcal{Y}$ where the randomness is independent of (X, Y), we have that $I(X; \theta(Y)) \leq I(X; Y)$. We have equality if and only if $\theta$ is an injective map.
\end{lemma}

\begin{lemma}[Strict Conditional Data Processing Inequality \citep{schoenebeck2020two}]\label{dataproc2} If $(W, X, Y)$ on a finite space $\mathcal{W} \times \mathcal{X} \times \mathcal{Y}$ is second-order stochastic relevant (i.e., for any distinct $x, x' \in \mathcal{X}$, there is some $w \in \mathcal{W}$ such that $P(Y|w, x) \neq P(Y|w, x')$) and has full support, then for any random function $\theta: \mathcal{Y} \rightarrow \mathcal{Y}$ where the randomness is independent of (W, X, Y), we have that $I(X; \theta(Y) | W) \leq I(X; Y | W)$. We have equality if and only if $\theta$ is an injective map.

\end{lemma}

These lemmas are the cornerstone of the mutual information paradigm, as they assert that whenever a mechanism pays an individual proportional to the mutual information between their signal and a peer's, the individual is incentivized to be truthful. Any strategy that involves lying, i.e., post-processing the signal, will only result in lower reward, and is hence disincentivized. In this work,  our result on truthfulness of the review mechanism under our payment scheme (Theorem \ref{truththm}) relies on a similar argument. We also note that a reviewer's review-based payment (i.e., their `target' payment) is a conditional mutual information \emph{in every equilibrium}.

\section{Relationship to HMIM and TDPP}\label{relsec}

Our two main considerations are incentivizing reviewers to exert effort, and to report their true reviews. We achieve both by drawing on ideas from the Target Differential Peer Prediction (TDPP) \citep{schoenebeck2020two} and Hierarchical Mutual Information Paradigm (HMIP)  \citep{kong2018eliciting}. The Differential Peer Prediction mechanism is a recently proposed prediction market-style mechanism applicable to our small-group, heterogeneous reviewer setting and determines payments for reviews of each paper independently. They present their mechanism for a group of 3 agents, ideal for our peer review setting. The mechanism works as follows: first, it assigns the 3 agents the roles of \emph{expert}, \emph{target}, and \emph{source}. Then, it asks the \emph{expert} for a probability distribution representing their prediction of the \emph{target's} signal. Finally, it reveals the \emph{source's} signal and asks the expert for an updated prediction of the target's signal. \citet{schoenebeck2020two} show that the total expected payment to agents is the mutual information of the source and target's signals conditioned on the expert's signals, and this yields a strong truthfulness guarantee \citep{kong2019information}. 

However, truthfulness is not enough; reporting low-effort signals can still be truthful. Additionally, since these truthful signals are paid their mutual information, when low-effort signals correlate more than high-effort signals, it becomes natural to coordinate into less informative equilibria. HMIP proposes a framework to get around this issue by constructing a hierarchy of criteria that range from requiring low-effort to high-effort signals, and paying agents the \emph{gain} in information on higher level criteria. Single-HMIM is an instantiation of the HMIP framework to incentivize effortful reviews on a single task, but uninformative equilibria can pay more than truthful equilibria and we seek to overcome this.

We combine the insights from TDPP and HMIP to propose a Hierarchical-Differential Peer Prediction (H-DIPP) mechanism.

\subsection{Relationship to HMIM}

Here we discuss how our mechanism differs from the single-HMIM proposal under the HMIP framework \citep{kong2018eliciting}.
\begin{itemize}
    \item {\bf Compulsory Reports}: Unlike single-HMIM, reviewers must provide scores for all criteria, not just the ones they have completed. This is currently the norm in peer review, and we would like to design the mechanism around this expectation. We model reviewers' alternative to effortful reports as `best guess' reports, as opposed to HMIM where a reviewer who does not complete a criterion does not report any signal and simply forgoes any payment on it. Realistically, a reviewer's signals on completed criteria may enable them to make reasonable guesses on higher level criteria, and we account for this in computing reviewers' marginal payments and setting hyperparameters.
    \item {\bf Partial Vs Total Order}: \citet{kong2018eliciting} only impose a partial order on their criteria hierarchy to keep it general; we impose a total order for ease of exposition.
    \item {\bf Different Payment Scheme}: While we adopt the hierarchical criteria structure from HMIP, the actual payments we provide are adopted from TDPP. However, at equilibrium, a reviewer's expected target payment is almost the mutual information payment scheme proposed in HMIP, under our conditional independence assumption (Assumption \ref{condindep}). The difference is that while HMIP (under the Conditional Independence assumption) would compute the mutual information of the target's and source's signals on criterion $t$ conditioned on the source's signals on lower level criteria, our mechanism also conditions this on the expert's private information.
    \item {\bf Strategy and Equilibria}: Our mechanism, like single-HMIM, is strictly truthful. The `idealized' HMIP has the desirable property that truthful reporting is a dominant strategy, but the practical single-HMIM algorithm does not have this property. Truthful reports are not a dominant strategy in our H-DIPP mechanism either, since target payments depend on the prediction strategy of experts. However, it is the case that a reviewer is paid greater target payments in truth-telling equilibria, and the use of proper scoring rules incentivizes experts to predict truthfully if targets deviate to truth-telling.
    
\end{itemize}

\subsection{Relationship to TDPP}

Here we discuss how our mechanism differs from the TDPP proposal \citep{schoenebeck2020two}.
\begin{itemize}

\item {\bf Playing All Roles}: TDPP assigns each agent a single role (expert, target, or source). In our mechanism, every agent plays all three roles. 

\item {\bf Effort Model}: TDPP assumes that signals are revealed to agents at no cost. We assume that the effort level determines the signals that are revealed to agents, so we adopt the hierarchical approach of HMIP to incentivize the effort needed to respond to higher level criteria.

\item {\bf Knowledge of Common Prior}: TDPP and other `detail-free' mechanisms assume no knowledge of the common prior. While we do not strictly need to know a reviewer triplet's common-prior to run the mechanism, we would like to know certain properties (the marginal payments from Lemma \ref{margrewlemma}) of the common prior (as in HMIP) so that we can set the mechanism hyperparameters to incentivize the desired fully informative equilibrium.

\item {\bf Multi-criteria Predictions and No Source Payment}: TDPP only elicits source and target responses on a single question, while H-DIPP elicits responses to a hierarchy of criteria. H-DIPP reveals the source's signal after each criterion (as opposed to revealing it all at once, or revealing target signals, etc.) as this brings the expected target payment close to the HMIP recommendation. However, this does mean that unlike TDPP, we do not provide any source payment for the following reason. The TDPP expected source payment was based on the expert $k$'s first prediction of target $k'$'s signals on a criterion $t$ (made without the source's information). This term did {not} change the source's incentives in the original TDPP (as their actions could not affect the expert's first prediction), but it was a convenient accounting trick to make the ex ante agent welfare equal to the conditional mutual information of the source and target's reports. However, we cannot incorporate such a source payment in H-DIPP: the expert makes predictions with the source's signals on lower level criteria, so \emph{the source can influence the expert's first prediction on every criterion $t>1$}. Thus, including such a source payment could lead the source to strategically misreport so that the expert's first predictions are poor. Naturally, leaving out source payments preserves source's incentives to report truthfully, but the ex ante reviewer welfare is no longer the conditional mutual information. This is why there are `extra' entropy terms in Lemma \ref{revwelf}.
    
    \item {\bf Strategy and Equilibria}: Like TDPP, our H-DIPP mechanism is strictly truthful. However, TDPP has the additional very desirable property of being \emph{strongly truthful}, i.e., the truth-telling equilibrium pays more than all other equilibria, so agents have no reason to coordinate into other equilibria. Owing to the multi-round structure of H-DIPP, we cannot quite make the same guarantee; it \emph{may} be possible to have equilibria where reviewers strategically misreport to increase a peer's target payment by more than the loss of their target payment thereby increasing aggregate payment, although there is no individual incentive to do so. It may be possible for all reviewers to do this so their individual and aggregate payments are higher than the fully informative equilibrium, although accepting the increased target payment bestowed by a source's strategic misreports and \emph{not reciprocating} by reporting truthfully would lead to a higher-paying equilibrium for the target. Investigating such equilibria is an important direction for future work. Nevertheless, we can make the weaker claim that the fully informative equilibrium pays more than an effortless uninformative equilibrium (Theorem \ref{ir}).
\end{itemize}

\section{Other Assumptions}\label{otherass}

We discuss some of the other assumptions underlying our H-DIPP review mechanism (Stage II) here.

\paragraph{Strategic and Uninformed Reports} Our mechanism will elicit reviewers' private review scores $X_k^t$ on all criteria, and we denote the (potentially strategic) \emph{reported review score} as $\hat{X}_k^t$. With $\theta_k$ as the transition matrix describing reviewer $k$'s reporting strategy, the joint distribution of true signals and reported signals under a signal strategy profile $(\theta_A, \theta_B ,\theta_C)$ is $\mathds{Q}(\hat{\bf X}_A^{[T]}, \hat{\bf X}_B^{[T]}, \hat{\bf X}_{C}^{[T]}, {\bf X}_A^{[T]}, {\bf X}_B^{[T]}, {\bf X}_{C}^{[T]}) = \theta_A(\hat{\bf X}_A^{[T]} | {\bf X}_A^{[T]}) \theta_B(\hat{\bf X}_B^{[T]} | {\bf X}_B^{[T]}) \theta_C(\hat{\bf X}_C^{[T]} | {\bf X}_C^{[T]}) \mathds{P}({\bf X}_A^{[T]}, {\bf X}_B^{[T]}, {\bf X}_{C}^{[T]})$.  When a reviewer plays a truthful reporting strategy (i.e., $\theta_k$ is identity), we denote it is as $\tau_k$.

Now, if a reviewer only completes up to $t$ criteria, they will not arrive at an informed private score on criteria at higher levels. However, their signals on lower level criteria do provide some information, allowing them to make reasonable guesses of the scores on higher level criteria\footnote{This is an important difference between our approach and the original HMIP, where reporting on incomplete criteria is optional. See Appendix \ref{relsec} for more details.}. Indeed, we assume that on incomplete higher level criteria, a reviewer's default private signal is their `best guess' of the signal that other agents would expect of them given the information they already have, and this may depend on how much effort the other agents exert. This is given by some function $g_{k}^{t'|t}: \Sigma_1 \times \ldots \times \Sigma_{t} \rightarrow \Sigma_{t'}$ where $x_k^{t'} = g_k^{t'|t}({\bf x}_k^{[t]})$ 
represents reviewer $k$'s best guess of their private assessment on criterion $t'$ when they have completed only criteria at level $t$ (where ${t'} > t)$. Note that $x_k^t$ thus either refers to an effort-informed signal drawn from the common prior or a best-guess, depending on the reviewer's effort level. We use the notation ${\bf x}_k^{[t_k]}$ if we wish to emphasize that we are only referring to the effort-informed signals on criteria $[t_k]$  (when $t_k < T$).

  \paragraph{Conditional Independence and Stochastic Relevance} First, we adopt a conditional independence assumption from \citet{kong2018eliciting}, which states that a reviewer $k'$'s private score on some criterion $t$ contains all the information needed to predict their peer $k$'s score on that criterion, given that peer's scores on lower level criteria. We write $I(\cdot)$ to represent the mutual information (see Appendix \ref{mipsec} for more details).

   \begin{assumption}[Conditional Independence]\label{condindep} 
  We make the following conditional independence assumption on the common prior:
  \begin{equation}
      I({\bf X}_{k'}^{[T]}, X_k^t | {\bf X}_k^{[t-1]}) = I(X_{k'}^{t}, X_k^t | {\bf X}_k^{[t-1]})
  \end{equation}
  or equivalently,
  \begin{equation}
      I( {\bf X}_{k'}^{[T]}\backslash X_{k'}^t, X_k^t | {\bf X}_k^{[t-1]}, X_{k'}^t) = 0
  \end{equation}
  \end{assumption}
  
 Finally, we make an assumption on the informativeness of the signals a reviewer might observe. When a reviewer $k$ arrives at review scores ${\bf x}_k^{[t]}$ on $[t]$ criteria, they also arrive at a posterior belief  $\mathds{P}({\bf  X}_{-k}^{[T]}|{\bf x}_k^{[t]})$ about the review scores of other reviewers. Informally, we assume there is some signal a reviewer could observe, which when paired with two different signals from a second reviewer, would induce two different beliefs about the signals of the third reviewer.

  \begin{assumption}[Second Order Stochastic Relevance] \label{stoch} 
   Suppose that for any distinct signals $x_{B}^t, \tilde{x}_{B}^t \in \Sigma_t$ on any criterion $t$, there is $x_{A}^t \in \Sigma_t$ such that: $$\mathds{P}(X_{C}^t | x_{A}^t, x_{B}^t) \neq \mathds{P}(X_{C}^t | x_{A}^t, \tilde{x}_{B}^t)$$ 
   We assume the above holds for any permutation of the reviewers $\{A, B, C\}$, so the common prior $\mathds{P}$ is second order stochastic relevant.
  \end{assumption}

\section{Additional Results}\label{app:add}
 
 Here, we give some additional results characterizing the marginal payments to reviewers, optimal setting of mechanism hyperparameters, and aggregate reviewer payments. We also give a sample walkthrough for our H-DIPP reviewing mechanism.
 
Going forward, we use $R$ to denote expected payment. A reviewer's \emph{marginal expected payment} is the additional payment they gain for privately exerting enough effort to complete an additional criterion, under a fixed strategy by the other agents. We give the following lemma on reviewers' marginal expected payments for completing an additional criterion and show that it is always non-negative (when peers have completed all criteria and are reporting truthfully). We break up the total marginal expected payment $\Delta R_k$ into the marginal expected expert and target payments $\Delta R_k^\text{expert}$ and $\Delta R_k^\text{target}$, which in turn are each broken up into marginal payments $\Delta R_{k^t}^\text{expert}$ and $\Delta R_{k^t}^\text{target}$ on each criterion.

\begin{lemma}[Non-negative Marginal Payments]\label{margrewlemma}
Suppose reviewers $B, C$ complete all criteria and truthfully report their private review scores and predictions. Then, reviewer $A$'s marginal expected payment for completing criterion $t_A$ and honestly reporting private review scores and predictions is $\Delta R_A({t_A}) := R_A(t_A) - R_A(t_A-1)   =  \Delta {R}_A^{\text{expert}}(t_A) + \Delta{R}_A^{\text{target}}(t_A)$, where:
\begin{enumerate}\small
\item $\Delta R_A^{\text{expert}}(t_A) := {R}_A^{\text{expert}}(t_A) - {R}_A^{\text{expert}}(t_A - 1) = \sum_{t=1}^T \alpha_A^t \Delta R_{A^t}^{{\text{ expert}}}(t_A)$
\item $\Delta R_A^{\text{target}}(t_A) := {R}_A^\text{target}(t_A) - {R}_A^{\text{target}}(t_A - 1)=\sum_{t=1}^T \beta_A^t \Delta R_{A^t}^\text{target}(t_A)$
    \item $\Delta R_{A^t}^{\text{expert}}(t_A):=  I \left(X_B^t; X_A^{t_A}|{\bf X}_A^{[t_A-1]}, {\bf X}_C^{[t]}\right)  + I\left(X_B^t; X_A^{t_A} | {\bf X}_A^{[t_A-1]}, {\bf X}_C^{[t-1]}\right)$
    \item $\Delta R_{A^t}^\text{target}(t_A) := \begin{cases}
    0, & t < t_A \\
     I\left({X}_B^{t_A}; {X}_A^{t_A} | {\bf X}_C^{[T]}, {\bf X}_B^{[t_A-1]}\right) - I\left({X}_B^{t_A}; g_A^{t_A|t_A-1}({\bf X}_A^{[t_A-1]}) | {\bf X}_C^{[T]}, {\bf X}_B^{[t_A-1]}\right), & t=t_A \\ 
     \mathds{E}\left[\text{KL}\left({\mathds{P}({X}_B^t| g_A^{t|t_A}({\bf X}_A^{[t_A]}), {\bf X}_C^{[T]}, {\bf X}_B^{[t-1]}) || \mathds{P}({X}_B^t| g_A^{t|t_A-1}({\bf X}_A^{[t_A-1]}), {\bf X}_C^{[T]}, {\bf X}_B^{[t-1]})}\right) \right], & t > t_A
    \end{cases}$
\end{enumerate}

Furthermore, the marginal expected expert and target payments on each criterion are always non-negative, i.e., $\Delta R_{A^t}^\text{expert}(t_A), \Delta R_{A^t}^\text{target}(t_A) \geq 0$ for all $t, t_A \in [T]$ and the inequality is strict when $t = t_A$. Similar results hold for reviewers $B$ and $C$ when peers complete all criteria and report truthfully.

\end{lemma}

With this, we can now specify how to set the mechanism hyperparameters to  incentivize a reviewer to complete all criteria, such that all reviewers completing all criteria and reporting truthfully is a strict BNE (equivalent to the \emph{potent} property of hyperparameters in \citet{kong2018eliciting}). 

\begin{theorem}[Optimal Mechanism Hyperparameters]\label{optparam}
Suppose mechanism hyperparameters $\{  \boldsymbol{\alpha}_A^{[T]}, \boldsymbol{\beta}_A^{[T]}, \boldsymbol{\alpha}_B^{[T]}, \boldsymbol{\beta}_B^{[T]}, \boldsymbol{\alpha}_C^{[T]}, \boldsymbol{\beta}_C^{[T]}\}$ are solutions to the following linear program:
\begin{equation}\small
\begin{split}
    &\min \sum_{k \in \{A, B, C \}} \sum_{t=1}^T \alpha_k^t R_{k^t}^\text{expert}(T) +  \beta_k^t R_{k^t}^\text{target}(T)\\
    &\text{s.t. } \sum_{t=1}^T \alpha_A^t \Delta R_{A^t}^\text{expert}(t_A) + \beta_A^t\Delta R_{A^t}^\text{target}(t_A) > \Delta e_A(t_A) \text{ for } 1 \leq t_A \leq T \\
    &\phantom{\text{s.t. }} \sum_{t=1}^T \alpha_B^t \Delta R_{B^t}^\text{expert}(t_B) + \beta_B^t\Delta R_{B^t}^\text{target}(t_B) > \Delta e_B(t_B) \text{ for } 1 \leq t_B \leq T \\
    &\phantom{\text{s.t. }} \sum_{t=1}^T \alpha_C^t \Delta R_{C^t}^\text{expert}(t_C) + \beta_C^t\Delta R_{C^t}^\text{target}(t_C) > \Delta e_C(t_C) \text{ for } 1 \leq t_C \leq T \\
    &\phantom{\text{s.t. }} \alpha_k^t, \beta_k^t \geq 0
\end{split}
\end{equation}

Then, for any reviewer $k$, if their peers complete all criteria and report predictions and review scores truthfully, reviewer $k$ strictly maximizes their expected utility by also completing all criteria and reporting predictions and review scores truthfully. In other words, completing all criteria and reporting truthfully is a strict Bayesian-Nash equilibrium. We refer to this equilibrium as the fully-informative equilibrium.
\end{theorem}

Having previously shown that our desired strategy profile (all reviewers completing all criteria and reporting truthfully) is a strict BNE, we now discuss how this compares to other equilibria and strategy profiles. A common challenge is that truthful equilibria pay less than effortless uninformative equilibria (like in single-HMIM), so the truthful equilibria are not realized in practice. However, this is not the case with H-DIPP; as we stated in Theorem \ref{ir} an uninformative equilibrium pays nothing in expectation and the fully informative equilibrium pays strictly more. To show this, we need the following lemma on the total payment to reviewers in any equilibrium:
 
 \begin{lemma}[Aggregate Reviewer Payments]\label{revwelf}
 Let $s_k = \left(e_k(t_k), \hat{\bf x}_k^{[T]}, \hat{\bf p}^+_{k'^{[T]} \leftarrow k}, \hat{\bf p}_{k'^{[T]} \leftarrow k}\right)$ be reviewer $k$'s strategy, where $k'$ is reviewer $k$'s target. If ${\bf s} = (s_A, s_B, s_C)$ is an equilibrium strategy profile,
 reviewer $A$'s expected payment is:
 \begin{equation} \label{indivrew} \small
     \mathds{E}[r_A] = \sum_{t=1}^T \alpha_A^t\left( -\mathds{E}\left[H\left(\mathds{Q}\left({\hat{X}_B^t | {\bf x}_A^{[t_A]}, \hat{\bf x}_C^{[t]}}\right)\right)\right] -\mathds{E}\left[H\left(\mathds{Q}\left({\hat{X}_B^t | {\bf x}_A^{[t_A]}, \hat{\bf x}_C^{[t-1]}}\right)\right)\right]\right)   + \beta_A^t I\left(\hat{X}_B^t; \hat{ X}_A^t | {\bf X}_C^{[t_C]}, \hat{\bf X}_B^{[t-1]}\right)
 \end{equation}
 
(and expected payments for peers $B,C$ computed analogously). Reviewers' total expected payment is:
\begin{equation}\small
    \begin{split}
    \mathds{E}[r_A + r_B + r_C] &=  \sum_{t=1}^T \alpha_A^t\left( -\mathds{E}\left[H\left(\mathds{Q}\left({\hat{X}_B^t | {\bf x}_A^{[t_A]}, \hat{\bf x}_C^{[t]}}\right)\right)\right] -\mathds{E}\left[H\left(\mathds{Q}\left({\hat{X}_B^t | {\bf x}_A^{[t_A]}, \hat{\bf x}_C^{[t-1]}}\right)\right)\right]\right)   + \beta_A^t I\left(\hat{X}_B^t; \hat{ X}_A^t | {\bf X}_C^{[t_C]}, \hat{\bf X}_B^{[t-1]}\right) \\
        &\phantom{=} + \alpha_B^t\left( -\mathds{E}\left[H\left(\mathds{Q}\left({\hat{X}_C^t | {\bf x}_B^{[t_B]}, \hat{\bf x}_A^{[t]}}\right)\right)\right] -\mathds{E}\left[H\left(\mathds{Q}\left({\hat{X}_C^t | {\bf x}_B^{[t_B]}, \hat{\bf x}_A^{[t-1]}}\right)\right)\right]\right)   + \beta_B^t I\left(\hat{X}_C^t; \hat{ X}_B^t | {\bf X}_A^{[t_A]}, \hat{\bf X}_C^{[t-1]}\right) \\
        &\phantom{=} + \alpha_C^t\left( -\mathds{E}\left[H\left(\mathds{Q}\left({\hat{X}_A^t | {\bf x}_C^{[t_C]}, \hat{\bf x}_B^{[t]}}\right)\right)\right] -\mathds{E}\left[H\left(\mathds{Q}\left({\hat{X}_A^t | {\bf x}_C^{[t_A]}, \hat{\bf x}_B^{[t-1]}}\right)\right)\right]\right)   + \beta_C^t I\left(\hat{X}_A^t; \hat{ X}_C^t | {\bf X}_B^{[t_B]}, \hat{\bf X}_A^{[t-1]}\right)
    \end{split}
\end{equation}
 \end{lemma}
 
 We make several observations. First, unlike the ideal-HMIP proposal (and like the concrete single-HMIM proposal) in \cite{kong2018eliciting}, we cannot claim that truthful reporting of private signals is a dominant strategy, since a reviewer's target payment for reporting private scores will depend on their peer's expert prediction strategy. However, it is the case that when we compare equilibria fixing (truthful or untruthful) peer strategies, a reviewer earns higher target rewards in the equilibria where they report truthfully, compared to ones where they strategically manipulate the score, since $I\left(\hat{X}_B^t; \hat{ X}_A^t | {\bf X}_C^{[t_C]}, \hat{\bf X}_B^{[t-1]}\right) \leq I\left(\hat{X}_B^t; { X}_A^t | {\bf X}_C^{[t_C]}, \hat{\bf X}_B^{[t-1]}\right)$ for all $t \in [T]$ (and similarly for other reviewers) due to the data processing inequality \citep{kong2019information}. % \jamie{I think this phrasing is a little confusing. I think you're saying it's a best response to truthfully report when other people truthfully report?}
 
 Second, unlike TDPP (which does not account for effort), we cannot guarantee the strong truthfulness of the fully informative truth-telling equilibrium, i.e., we cannot guarantee that the fully informative equilibrium has higher aggregate utility than any other equilibrium. It may be possible for a reviewer to strategically manipulate their private signals so that when used as a source for some expert's prediction, the target receives higher payment; e.g., it may be possible that reviewer $B$'s strategic reports provides a higher target payment for reviewer $A$ with $I\left({X}_B^t; \hat{ X}_A^t | {\bf X}_C^{[t_C]}, {\bf X}_B^{[t-1]}\right) \leq I\left(\hat{X}_B^t; \hat{ X}_A^t | {\bf X}_C^{[t_C]}, \hat{\bf X}_B^{[t-1]}\right)$. Reviewer $B$'s strategic reports could then increase reviewer $A$'s target payment by more than it hurts their own target payment, thereby raising aggregate utility (costs are the same). However, reviewers do not receive payments for their role as source, and further must accept lower target payment for such strategic manipulation, so reviewers are disincentivized from pursuing such equilibria unilaterally. Nevertheless, we cannot rule out the possibility of viable collusive equilibria. %presence of equilibria where all reviewers lie strategically to raise another reviewer's payment. \jamie{Perhaps here this could just be referred to as lacking collusion resistance?}

 \subsection{Sample Mechanism Flow}
 
 Here, we walkthrough an example of how the H-DIPP reviewing mechanism would work according to Algorithm \ref{algo:duplicate2}. We emphasize that this is just an example; the number of questions, the response scale, etc. may be modified in practice.
 
 Suppose reviewers Alice, Bob, and Charlie are assigned to review paper $\omega$. Each reviewer reads and scores the paper on the same five pre-defined criteria on a five-point scale: 1) \emph{(presentation)} ``How clear is the writing/exposition?''; 2) \emph{(completeness)} ``How easy would it be to reproduce major results?''; 3) \emph{(correctness)} ``How accurate are the technical claims and methodology?''; 4) \emph{(contribution)} ``How interesting and valuable of a contribution is the work? ''; 5) \emph{(overall quality)} ``In which quintile of submitted papers would this work fall?" 
 
Once all the reviews are in, each of the reviewers logs into the review website and plays an asynchronous prediction game. We consider the game from Alice's perspective, who is tasked with predicting Bob's review, and shown Charlie's scores for assistance. Bob will be tasked with predicting Charlie's scores, with Alice's scores for assistance, and Charlie will be tasked with predicting Alice's scores, with Bob's scores for assistance. The game from Alice's perspective proceeds as follows:

\begin{enumerate}
    \item Q1: Presentation
    \begin{itemize}
        \item Alice makes a prediction of Bob's score on Q1: presentation.
        \item Alice is shown Charlie's score on Q1: presentation.
        \item Alice makes an updated prediction of Bob's score on Q1: presentation.
    \end{itemize}
    \item Q2: Completeness
    \begin{itemize}
        \item Alice makes a prediction of Bob's score on Q2: completeness.
        \item Alice is shown Charlie's score on Q2: completeness.
        \item Alice makes an updated prediction of Bob's score on Q2: completeness.
    \end{itemize}
    \item Q3: Correctness
    \begin{itemize}
        \item Alice makes a prediction of Bob's score on Q3: correctness.
        \item Alice is shown Charlie's score on Q3: correctness.
        \item Alice makes an updated prediction of Bob's score on Q3: correctness.
    \end{itemize}
    \item Q4: Contribution
    \begin{itemize}
        \item Alice makes a prediction of Bob's score on Q4: contribution.
        \item Alice is shown Charlie's score on Q4: contribution.
        \item Alice makes an updated prediction of Bob's score on Q4: contribution.
    \end{itemize}
    \item Q5: Overall Quality
    \begin{itemize}
        \item Alice makes a prediction of Bob's score on Q5: overall quality.
        \item Alice is shown Charlie's score on Q5: overall quality.
        \item Alice makes an updated prediction of Bob's score on Q5: overall quality.
    \end{itemize}
\end{enumerate}

The game then concludes for Alice. Once Bob and Charlie have also completed the game, the mechanism may determine payments as per the H-DIPP payment scheme.

  \section{Demand and Supply} \label{demsup}

In Section \ref{stage1}, we assumed the VCG auction had a fixed number of review slots $P$ to be auctioned. Here, we present the considerations and challenges in designing a mechanism that sets $P$ dynamically \emph{after} paper submissions based on the supply of reviewers, but a comprehensive solution is a task for future work. The foremost considerations in setting $P$ based on the availability of reviewers are:

\begin{enumerate}
    \item {\bf Individual Rationality and Truthfulness}: Authors and reviewers should not expect to lose credits by participating the mechanism, and authors should not be able to game the number of review slots by bidding something other than their expected value for having their paper reviewed.
    \item {\bf Balanced Budget}: We aim to set $P$ based on the number of reviews that can be paid for under the H-DIPP payment scheme using revenue from the auction. Recall that the cost to the mechanism to have one paper reviewed by \emph{a given reviewer triplet} in the fully-informative equilibrium (if this is indeed the observed equilibrium) is given by Lemma \ref{revwelf}, where hyperparameters have been set using Theorem \ref{optparam}. 
    
    \item {\bf Maximize $P$}: We would like to accept as many papers for review as feasible.
    
    \item {\bf Fair Expertise-based Matching}: Reviewers must be well-qualified to review their assigned papers and have the appropriate domain expertise. Additionally, all papers accepted for review should have equal consideration when being assigned reviewers; the bid accompanying a paper should have no bearing on the quality of reviewer matches\footnote{Aside from the motivation of fairness, if higher bidding papers received priority in matching reviewers, reviewers could infer a paper's bid from the quality of their match, giving them a low-effort signal of paper quality.}. Thus, although the review slots are identical from authors' perspective (for fairness), the mechanism must match them to reviewers appropriately. We assume access to a black-box algorithm that can score a reviewer's match with a paper, and assign reviewers to papers based on this score (e.g. Toronto Paper Matching System \citep{charlin2013toronto}). 
\end{enumerate}

 The final consideration comes from the observation that the mechanism should not blindly seek to minimize its total cost over all papers in an attempt to squeeze in as many review acceptances $P$ as possible. The cost to the mechanism for reviewing a single paper depends both on the hyperparameters (optimally set via Theorem \ref{optparam}) as well as the mutual information in reviewers' signals, so total costs can be made abnormally low by badly matching reviewers and paying them little since their signals will have naturally low mutual information. Heuristically, we expect that well-matched expert reviewers hit the sweet spot of having low effort costs to review a paper and more mutually informative review scores with expert peers, so the mechanism pays a fair price for such expert reviews. 
 
Nevertheless, we face a trade-off where accepting more papers for review can decrease the average quality of the paper-reviewer matches. % \jamie{also the amount you can charge for getting reviewed will decrease.} 
Formally, let $S_{i_k}^*$ denote the match score for the $k$-th reviewer of paper $\omega_i$ and $R_{i_k}^*$ denote the expected payment to the $k$-th reviewer of paper $\omega_i$ under the optimal assignment of reviewers (as determined by the black-box algorithm) for some given number of papers $P$ accepted for review. Then, a naive mechanism designer's objective is to pick the number of review slots to be the largest $P$ with an acceptable average reviewer match score that also satisfies the balanced budget constraint (where slots are priced as in the VCG auction):
\begin{equation}\small
    \begin{split}
        &\underset{P}{\text{argmax}} \,\,\,\, \left(\frac1{P} \sum_{i=1}^P S_{i_A}^* + S_{i_B}^* + S_{i_C}^* \right) + \lambda P\\
    &\text{s.t. } \,\,\,\, b^*_{P+1} \cdot P \geq \sum_{i=1}^P R_{i_A}^* + R_{i_B}^* + R_{i_C}^*
    \end{split}
\end{equation}

The parameter $\lambda$ controls the tradeoff between match quality and number of papers accepted for review, and is a design decision to be made by individual conferences. The objective can be modified if desired; substituting the mean reviewer match score for the median or an additional term penalizing variance in reviewer match scores are alternative choices. In a naive algorithm, we can find the optimal number of review slots by first computing all the valid number of review slots under the balanced budget constraint, and picking the largest one with an acceptable average reviewer match score (summarized in  Algorithm \ref{algo:demsup}).

\begin{algorithm}[h]
\DontPrintSemicolon
\KwIn{$N$ papers submitted to the conference, a black-box algorithm matching reviewers to papers based on expertise (e.g. Toronto Paper Matching System)}
\KwOut{Optimal Number of Review Slots $P$} 
\For{$n = 1:N$}{
Obtain optimal matching of reviewers to $n$ papers from black-box algorithm \\
If accepting $n$ papers satisfies the budget constraint $b^*_{P+1} \cdot P \geq \sum_{i=1}^P R_{i_A}^* + R_{i_B}^* + R_{i_C}^*$, add $p$ to the list of allowed number of review slots $\mathcal{P} \leftarrow \mathcal{P} \cup \{n\}$
}
Select the largest $P \in \mathcal{P}$ with an acceptable average reviewer match score, i.e., largest $P$ that maximizes the objective $\left(\frac1{P} \sum_{i=1}^P S_{i_A}^* + S_{i_B}^* + S_{i_C}^* \right) + \lambda P$
\caption{Naive Algorithm to Determine Number of Review Slots $P$}
\label{algo:demsup}
\end{algorithm}

However, such a mechanism faces several challenges. Firstly, it is not obvious that it is truthful; authors may be able to affect the number of review slots by changing their bids. Secondly, we need to identify the mechanism hyperparameters by solving the linear program in Theorem \ref{optparam} and computing the expected payment in the fully informative equilibrium (Lemma \ref{revwelf}); this requires estimating or eliciting information about the common prior $\mathds{P}$ and the reviewer triplet's effort cost functions $e_k(\cdot)$. Furthermore, in practice we may need a single set of mechanism hyperparameters $\{\boldsymbol{\alpha}^{[T]}, \boldsymbol{\beta}^{[T]} \}$ that determine payments to all reviewers. Lastly, when the number of review slots $P$ is fixed, authors' bids will not affect their utilities in the reviewing stage of the same conference since the hyperparameters will have been pre-determined; this allows us to analyze the two stages of the mechanism independently. However, endogenizing $P$ would mean the setting of hyperparameters (and hence payments to reviewers) is sensitive to the revenue raised (and hence bids made by authors). This coupling of the two stages complicates the analysis. These challenges are important directions for future work. %These hyperparameters could be set so that they satisfy the constraints of the linear program for the `most expensive' reviewer in the reviewer pool, and all lower-cost reviewers will enjoy surplus. However, this still requires estimating or truthfully eliciting information about reviewer triplets' common priors and reviewers' marginal costs and is an important direction for future work. 

% Secondly, endogenizing the number of review slots $P$ may pose some thorny theoretical issues. It may imply that adding or removing submissions could either increase or decrease the number of review slots are not monotonic in the number of submissions (e.g. removing a poorly-matched reviewer , potentially making the auction vulnerable to collusion. It may also couple the stage I and stage II utilities (e.g. if higher bids could influence the hyperparameters and thus payment rates to reviewers)

\section{Directions for Future Work} \label{conc}

We present some important directions for future work below:

\begin{enumerate}

\item {\bf Criteria Hierarchy}: What is the best construction of the criteria hierarchy? 

\item  {\bf Estimating Hyperparameters and Costs}: Setting mechanism hyperparameters and the number of review slots $P$ requires information about reviewers' marginal payments and costs, so identifying ways to learn \citep{liu2016learning} or elicit this information would be valuable.

\item {\bf Endogenizing the Number of Review Slots $P$}: How can we allow $P$ to be set based on the availability of reviewers, while guaranteeing an individually rational, truthful, and balanced budget mechanism that accepts as many well-matched papers as possible?

    \item {\bf Variance in Payments}: Is the variance in payments so large  that honest and effortful reviews get negative payments with unacceptable regularity? If so, can the addition of a few more reviewers reduce this variance?
    
     \item {\bf Improving Equilibria}: Can we leverage information across reviewing tasks \citep{dasgupta2013crowdsourced, agarwal2017peer} in an efficient way to further increase total payoffs in the truth-telling equilibrium?  How does the utility of untruthful equilibria compare to the fully-informative equilibrium?  How robust is the mechanism to collusion?

      \item {\bf Truthful, but Useful?}: H-DIPP is strictly \emph{truthful}, but whether or not this is useful is a different empirical question. If H-DIPP-incentivized review scores vary widely even with post-processing (i.e., truthful reviews end up being quite idiosyncratically subjective), or do not track other measures of paper quality, it may not be possible to extract a meaningful signal (criticism C4). Running small scale experiments with real money could shed light on whether the benefits of our proposed mechanism warrant the overhead.

\end{enumerate}

\section{Proofs}\label{proofs}

 \paragraph{Theorem \ref{truththm}}(Strictly Truthful)
 \emph{Suppose reviewers $A, B, C$ complete $[t_A], [t_B]$, and $[t_C]$ criteria respectively. Then, on shared complete criteria $[t'] = [\min(t_A, t_B, t_C)]$, the H-DIPP mechanism is strictly truthful, i.e., all reviewers honestly reporting their true predictions and private review scores on shared complete criteria $[t']$ is a strict Bayesian Nash equilibrium.}
 
\begin{proof}[Proof of Theorem \ref{truththm}]
We adapt the proof of strict truthfulness of TDPP by \citet{schoenebeck2020two} for H-DIPP. We show that agents reporting their honest predictions and scores on shared complete criteria is a strict Bayesian Nash equilibrium; importantly, payments for each criterion are independent and can be strategized about independently, hence we can reason strictly about strategies on the shared complete criteria ${\bf x}_k^{[t']} \subseteq {\bf x}_k^{[t_k]}$. Reviewer $A$'s private signals on incomplete criteria $t > t_A$ are ${\bf x}_A^{t} = g_A^{t|t_A}({\bf x}_A^{[t_A]})$ for some best guess function $g$, and similarly for reviewers $B$ and $C$.
Since all reviewers play the role of expert, target, and source, we only analyze reviewer $A$'s incentives when peers $B$ and $C$ report truthfully; these incentives are symmetrical for all reviewers. Additionally, we hold the number of completed criteria and hence effort cost $e_k(t_k)$ constant, and look to optimize payoff under this condition.

Suppose reviewers $B$ and $C$ report their private signals truthfully, i.e., $\hat{x}_B^t = x_B^t$ and $\hat{x}_C^t = x_C^t$, and reviewer $C$ as expert truthfully predicts reviewer $A$'s score as the posterior from updating the common prior $\hat{\bf p}_{A^t \leftarrow C} = \mathds{P}(X_A^t | {\bf x}_C^{[t']}, {\bf x}_B^{[t-1]})$, and $\hat{\bf p}^+_{A^t \leftarrow C} = \mathds{P}(X_A^t | {\bf x}_C^{[t']}, {\bf x}_B^{[t]})$.

Now, reviewer $A$'s expected payment (taken with respect to the joint distribution $\mathds{Q}(\hat{\bf X}^{[T]}_A, {\bf X}^{[t']}_A ,{\bf X}^{[t']}_B,{\bf X}^{[t']}_C)$ of true signals and $A$'s potentially strategic reports) on the shared complete criteria can be written as:

\begin{equation} \label{eq:totrew1}\small
\begin{split}
    \mathds{E}_{\mathds{Q}}[r_{A^t}]
    &= \alpha_A^t\mathds{E}\left[\underbrace{[S_\text{log}({x}_B^t, \hat{\bf p}_{B^t \leftarrow A}) + S_\text{log}({x}_B^t, \hat{\bf p}^+_{B^t \leftarrow A})}_{\text{expert payment}}\right] +  \beta_A^t\mathds{E}\left[\underbrace{S_\text{log}(\hat{x}_A^t, \mathds{P}(X_A^t | {\bf x}_C^{[t']}, {\bf x}_B^{[t]})) - S_\text{log}(\hat{x}_A^t, \mathds{P}(X_A^t | {\bf x}_C^{[t']}, {\bf x}_B^{[t-1]}))}_{\text{target payment}}\right] 
    \end{split}
\end{equation}

Now, reviewer $A$'s payment on a given criterion $t$ is independent of their reported signals and predictions on other criteria, so the sum of payments on complete criteria can be maximized by maximizing the expert payment and target payment on each criterion: 

\emph{Expert Payment} \space\space The expected expert payment consists of  reviewer $A$'s predictions of their target $B$'s signal on each criterion $t$ first with the source $C$'s signals on criteria less than $t$ and then with $C$'s signal on criterion $t$ as well. Predictions on each criterion are scored independently, so a strategic prediction on one criterion cannot impact future payments on other criteria. Since we use a strictly proper scoring rule to score both predictions on each criterion, reviewer $A$'s expected payoff for each criterion is strictly maximized when they report their true posterior belief over their target $B$'s signals under the common prior, i.e., $\hat{\bf p}^+_{B^t \leftarrow A} = \mathds{P}(X_B^t | {\bf x}_A^{[t_A]}, {\bf x}_C^{[t]})$ and $\hat{\bf p}_{B^t \leftarrow A} = \mathds{P}(X_B^t | {\bf x}_A^{[t_A]}, {\bf x}_C^{[t-1]})$.

\emph{Target Payment} \space\space Now, consider the reviewer $A$'s expected target payment on criteria $t \leq t'$. We  observe that we can model the strategies for responses to different criteria independently since the target payments are also independent for different criteria. Let $\hat{x}_A^t = \theta_{A^t}({\bf x}_A^{[t_A]})$ be reviewer $A$'s \emph{deterministic best response} report (on criterion $t$) to reviewer $B$'s honest signal reports $\hat{\bf x}_B^{[T]}$ and reviewer $C$'s honest (posterior) predictions $\hat{\bf p}^+_{A^t \leftarrow C} = \mathds{P}({X}_A^t | {\bf x}_C^{[t_C]}, {\bf x}_B^{[t]})$, $\hat{\bf p}_{A^t \leftarrow C} = \mathds{P}({X}_A^t | {\bf x}_C^{[t_C]}, {\bf x}_B^{[t-1]})$. 

For $t \leq t'$, reviewer $A$'s expected target payment for playing this strategy $\hat{x}_A^t =\theta_{A^t}({\bf x}_A^{[t']})$ is:

\begin{equation}\small
\begin{split}
    u_t(\theta) :&= \mathds{E}\left[S_\text{log}(\hat{x}_A^t, \mathds{P}(X_A^t | {\bf x}_C^{[t_C]}, {\bf x}_B^{[t]})) - S_\text{log}(\hat{x}_A^t, \mathds{P}(X_A^t | {\bf x}_C^{[t_C]}, {\bf x}_B^{[t-1]}))\right] \\ 
    &= \mathds{E}_{\mathds{Q}}\left[\log\frac{\mathds{P}(\hat{x}_A^t | {\bf x}_C^{[t_C]}, {\bf x}_B^{[t]})}{\mathds{P}(\hat{x}_A^t | {\bf x}_C^{[t_C]}, {\bf x}_B^{[t-1]})}\right] \\
    &= \mathds{E}_{\mathds{Q}}\left[\log\frac{\mathds{P}(\hat{x}_A^t, x_B^t | {\bf x}_C^{[t_C]}, {\bf x}_B^{[t-1]})}{\mathds{P}(\hat{x}_A^t | {\bf x}_C^{[t_C]}, {\bf x}_B^{[t-1]}) \mathds{P}(x_B^t|{\bf x}_C^{[t']}, {\bf x}_B^{[t-1]})}\right] \\
    % &= \mathds{E}_{\mathds{Q}}\left[\log\frac{\mathds{P}({x}_B^t | {\bf x}_C^{[t_C]}, {\bf x}_B^{[t-1]}, \hat{x}_A^t)}{\mathds{P}(x_B^t | {\bf x}_C^{[t_C]}, {\bf x}_B^{[t-1]})}\right] 
    &= \mathds{E}_{\mathds{P}}\left[\log\frac{\mathds{P}({x}_B^t | {\bf x}_C^{[t_C]}, {\bf x}_B^{[t-1]}, \theta_{A^t}({\bf x}_A^{[t_A]}
    ))}{\mathds{P}(x_B^t | {\bf x}_C^{[t_C]}, {\bf x}_B^{[t-1]})}\right] 
    \end{split}
\end{equation}

By a similar calculation, we have that reviewer $A$'s expected payment under a truth-telling strategy is $u_t(\tau):=\mathds{E}_{\mathds{P}}\left[\log\frac{\mathds{P}({x}_B^t | {\bf x}_C^{[t_C]}, {\bf x}_B^{[t-1]}, {x}_A^t)}{\mathds{P}(x_B^t | {\bf x}_C^{[t_C]}, {\bf x}_B^{[t-1]})}\right] $. Then, the difference between truth-telling payment and the best response payment is:
\begin{equation}
    \begin{split}
        u_t(\tau) - u_t(\theta)  &= \mathds{E}_{\mathds{P}}\left[\log\frac{\mathds{P}({x}_B^t | {\bf x}_C^{[t_C]}, {\bf x}_B^{[t-1]}, {x}_A^t)}{\mathds{P}({x}_B^t | {\bf x}_C^{[t_C]}, {\bf x}_B^{[t-1]}, \theta_{A^t}({\bf x}_A^{[t_A]}))}\right] \\
        &= \mathds{E}_{{\bf X}_A^{[t_A]}, {\bf X}_B^{[t-1]}, {\bf X}_C^{[t_C]}}\left[\mathds{E}_{X^{t}}\left[\log\frac{\mathds{P}({x}_B^t | {\bf x}_C^{[t_C]}, {\bf x}_B^{[t-1]}, {x}_A^t)}{\mathds{P}({\bf x}_B^t | {\bf x}_C^{[t_C]}, {\bf x}_B^{[t-1]}, \theta_{A^t}({\bf x}_A^{[t_A]}))}\right]\left| {\bf X}_A^{[t_A]} = {\bf x}_A^{[t_A]},  {\bf X}_B^{[t-1]} = {\bf x}_B^{[t-1]}, {\bf X}_C^{[t_C]} = {\bf x}_C^{[t_C]} \right.\right] \\
        &= \mathds{E}_{{\bf X}_A^{[t_A]}, {\bf X}_B^{[t-1]}, {\bf X}_C^{[t_C]}}\left[\text{KL}\left({\mathds{P}({x}_B^t | x_C^{[t_C]}, x_B^{[t-1]}, {x}_A^t)}||{\mathds{P}({x}_B^t | x_C^{[t_C]}, x_B^{[t-1]}, \theta_{A^t}({\bf x}_A^{[t_A]}))}\right)\right]
    \end{split}
\end{equation}

Since KL is a divergence, it is greater than or equal to zero, so we must have $u_t(\tau) - u_t(\theta) \geq 0$. Finally, observe that if the strategy $\theta$ is not truthful, i.e., $\theta_{A^t}({\bf x}_A^{[t_A]}) \neq {x}_A^t$, then by the second-order stochastic relevance there is some $x_C^{t} \subset {\bf x}_C^{[t_C]}$ for which ${\mathds{P}({x}_B^t | {\bf x}_C^{[t_C]}, {\bf x}_B^{[t-1]}, {x}_A^t)} \neq {\mathds{P}({x}_B^t | {\bf x}_C^{[t_C]}, {\bf x}_B^{[t-1]}, \hat{x}_A^t)}$, so the KL divergence is not identically zero over all signals in the expectation. Consequently, the inequality is strict, implying $u(\tau) > u(\theta)$. Thus, the best response strategy is truthtelling ($\theta=\tau$), and reviewer $A$ strictly maximizes their target payment by truthfully reporting their signal $\hat{x}_A^t = x_A^t$ on all complete criteria.

We have shown that when peers report their predictions and private review scores truthfully, a reviewer should truthfully report their prediction to strictly maximize their expert payment and truthfully report their private review score to strictly maximize their target payment. These incentives are symmetrical for reviewers $B$ and $C$, so truthful reporting is a strict Bayesian Nash Equilibrium.
\end{proof}

\paragraph{Lemma \ref{margrewlemma}}
\emph{Suppose reviewers $B, C$ complete all criteria and truthfully report their private review scores and predictions. Then, reviewer $A$'s marginal expected payment for completing criterion $t_A$ and honestly reporting private review scores and predictions is $\Delta R_A({t_A}) := R_A(t_A) - R_A(t_A-1)   =  \Delta {R}_A^{\text{expert}}(t_A) + \Delta{R}_A^{\text{target}}(t_A)$, where:}
\begin{enumerate}\small
\item $\Delta R_A^{\text{expert}}(t_A) := {R}_A^{\text{expert}}(t_A) - {R}_A^{\text{expert}}(t_A - 1) = \sum_{t=1}^T \alpha_A^t \Delta R_{A^t}^{{\text{ expert}}}(t_A)$
\item $\Delta R_A^{\text{target}}(t_A) := {R}_A^\text{target}(t_A) - {R}_A^{\text{target}}(t_A - 1)=\sum_{t=1}^T \beta_A^t \Delta R_{A^t}^\text{target}(t_A)$
    \item $\Delta R_{A^t}^{\text{expert}}(t_A):=  I \left(X_B^t; X_A^{t_A}|{\bf X}_A^{[t_A-1]}, {\bf X}_C^{[t]}\right)  + I\left(X_B^t; X_A^{t_A} | {\bf X}_A^{[t_A-1]}, {\bf X}_C^{[t-1]}\right)$
    \item $\Delta R_{A^t}^\text{target}(t_A) := \begin{cases}
    0, & t < t_A \\
     I\left({X}_B^{t_A}; {X}_A^{t_A} | {\bf X}_C^{[T]}, {\bf X}_B^{[t_A-1]}\right) - I\left({X}_B^{t_A}; g_A^{t_A|t_A-1}({\bf X}_A^{[t_A-1]}) | {\bf X}_C^{[T]}, {\bf X}_B^{[t_A-1]}\right), & t=t_A \\ 
     \mathds{E}\left[\text{KL}\left({\mathds{P}({X}_B^t| g_A^{t|t_A}({\bf X}_A^{[t_A]}), {\bf X}_C^{[T]}, {\bf X}_B^{[t-1]}) || \mathds{P}({X}_B^t| g_A^{t|t_A-1}({\bf X}_A^{[t_A-1]}), {\bf X}_C^{[T]}, {\bf X}_B^{[t-1]})}\right) \right], & t > t_A
    \end{cases}$
\end{enumerate}

\emph{Furthermore, the marginal expected expert and target payments on each criterion are always non-negative, i.e., $\Delta R_{A^t}^\text{expert}(t_A), \Delta R_{A^t}^\text{target}(t_A) \geq 0$ for all $t, t_A \in [T]$ and the inequality is strict when $t = t_A$. Similar results hold for reviewers $B$ and $C$ when peers complete all criteria and report truthfully.} 

\begin{proof}
 Assuming peers $B, C$ complete all $T$ criteria and report truthfully, we can write reviewer $A$'s  strategy profile of effort $e_A(t_A)$ and truthful reports as $s_A = \left(e_A(t_A), {\bf x}_k^{[t_A]} \cup \{g_A^{t'|t_A}({\bf  x}_A^{[t_A]})\}_{t' \in [t_A+1, T]}, \{\mathds{P}(X_B^t| {\bf X}_A^{[t_A]}, {\bf X}_C^{[t]}), \mathds{P}(X_B^t| {\bf X}_A^{[t_A]}, {\bf X}_C^{[t-1]})\}_{t \in [T]}\right)$, and the expected payment is:

\begin{equation} \small
\begin{split}
    {R}_A &= \underbrace{\sum_{t=1}^T \alpha_A^t \mathds{E}_{{\bf X}_A^{[t_A]}, X_B^t, {\bf X}_C^{[t]}}\left[ \log\left(\mathds{P}(X_B^t| {\bf X}_A^{[t_A]}, {\bf X}_C^{[t]})\right) +  \log\left(\mathds{P}(X_B^t| {\bf X}_A^{[t_A]}, {\bf X}_C^{[t-1]})\right)\right]}_{{R}_A^\text{expert}} \\
 &\phantom{=} + \underbrace{\sum_{t=1}^T \beta_A^t \mathds{E}_{X_A^{t}, {\bf X}_B^{[t]}, {\bf X}_C^{[T]}}\left[ \log\left(\frac{\mathds{P}(X_A^t|{\bf X}_C^{[T]}, {\bf X}_B^{[t]})}{\mathds{P}(X_A^t|{\bf X}_C^{[T]}, {\bf X}_B^{[t-1]})}\right)\right]}_{{R}_A^\text{target}} 
    \end{split}
\end{equation}

The expected payment $R_A$ is technically a function of all reviewers' completed criteria $R_A(t_A, t_B, t_C)$, but we simply write it as $R_A(t_A)$ as other reviewers' are assumed to have completed $T$ criteria. The marginal payment for completing criterion $t_A$ is $R_A(t_A) - R_A(t_A-1) = \Delta R_A({t_A})  = \Delta {R}_A^\text{expert}(t_A) + \Delta{R}_A^\text{target}(t_A)$. 
The marginal expert payment $\Delta{R}_A^\text{expert}(t_A)$ is:

\begin{equation}\label{margexp}\small
    \begin{split}
        \Delta{R}_A^\text{expert}(t_A) &= {R}_A^\text{expert}(t_A) - {R}_A^\text{expert}(t_A - 1) \\ 
        &=\sum_{t=1}^T \alpha_A^t \left(\mathds{E}_{{\bf X}_A^{[t_A]}, X_B^t, {\bf X}_C^{[t]}}\left[ \log\left(\mathds{P}(X_B^t|{\bf X}_A^{[t_A]}, {\bf X}_C^{[t]})\right) +  \log\left(\mathds{P}(X_B^t|{\bf X}_A^{[t_A]}, {\bf X}_C^{[t-1]})\right)\right]  \right) \\
        &\phantom{=} - \sum_{t=1}^T \alpha_A^t \left(\mathds{E}_{{\bf X}_A^{[t_A-1]}, X_B^t, {\bf X}_C^{[t]}}\left[ \log\left(\mathds{P}(X_B^t| {\bf X}_A^{[t_A-1]}, {\bf X}_C^{[t]})\right) +  \log\left(\mathds{P}(X_B^t| {\bf X}_A^{[t_A-1]}, {\bf X}_C^{[t-1]})\right)\right] \right) \\
        &= \sum_{t=1}^T  \alpha_A^t \left(\mathds{E}_{{\bf X}_A^{[t_A]}, X_B^t, {\bf X}_C^{[t]}}\left[ \log\left(\mathds{P}(X_B^t|{\bf X}_A^{[t_A]}, {\bf X}_C^{[t]})\right) -  \log\left(\mathds{P}(X_B^t| {\bf X}_A^{[t_A-1]}, {\bf X}_C^{[t]})\right)\right]  \right. \\
  &\phantom{=} \left. + \mathds{E}_{{\bf X}_A^{[t_A]}, X_B^t, {\bf X}_C^{[t-1]}}\left[ \log\left(\mathds{P}(X_B^t|{\bf X}_A^{[t_A]}, {\bf X}_C^{[t-1]})\right) -  \log\left(\mathds{P}(X_B^t|{\bf X}_A^{[t_A-1]}, {\bf X}_C^{[t-1]})\right)\right]  \right) \\
  &= \sum_{t=1}^T  \alpha_A^t \left(\mathds{E}_{{\bf X}_A^{[t_A]}, X_B^t, {\bf X}_C^{[t]}}\left[ \log\left(\frac{\mathds{P}(X_B^t|{\bf X}_A^{[t_A]}, {\bf X}_C^{[t]})}{\mathds{P}(X_B^t|{\bf X}_A^{[t_A-1]}, {\bf X}_C^{[t]})}\right)\right] + \mathds{E}_{{\bf X}_A^{[t_A]}, X_B^t, {\bf X}_C^{[t-1]}}\left[ \log\left(\frac{\mathds{P}(X_B^t| {\bf X}_A^{[t_A]}, {\bf X}_C^{[t-1]})}{\mathds{P}(X_B^t|{\bf X}_A^{[t_A-1]}, {\bf X}_C^{[t-1]})}\right)\right]  \right) \\
   &= \sum_{t=1}^T  \alpha_A^t \left(I \left(X_B^t; X_A^{t_A}|{\bf X}_A^{[t_A-1]}, {\bf X}_C^{[t]}\right)  + I\left(X_B^t; X_A^{t_A} | {\bf X}_A^{[t_A-1]}, {\bf X}_C^{[t-1]}\right)\right)
    \end{split}
\end{equation}

Defining $\Delta R_{A^t}^\text{expert}(t_A) := I \left(X_B^t; X_A^{t_A}|{\bf X}_A^{[t_A-1]}, {\bf X}_C^{[t]}\right)  + I\left(X_B^t; X_A^{t_A} | {\bf X}_A^{[t_A-1]}, {\bf X}_C^{[t-1]}\right)$, we have $\Delta R_A^{\text{expert}}(t_A) = \sum_{t=1}^T \alpha_A^t  \Delta R_{A^t}^{{\text{ expert}}}(t_A)$. Next, we compute the marginal target payment as follows (splitting the sum over $T$ criteria into three: marginal target payment on criteria less than $t_A$, criterion $t_A$, and criteria greater than $t_A$):

\begin{equation}\small
\begin{split}
  \Delta{R}_A^\text{target}(t_A) &= {R}_A^\text{target}(t_A) - {R}_A^\text{target}(t_A - 1) \\  
  &=  \sum_{t=1}^{t_A-1} \beta_A^t \left(\mathds{E}\left[ \log\left(\frac{\mathds{P}(X_A^t|{\bf X}_C^{[T]}, {\bf X}_B^{[t]})}{\mathds{P}(X_A^t|{\bf X}_C^{[T]}, {\bf X}_B^{[t-1]})}\right)\right] - \mathds{E}\left[ \log\left(\frac{\mathds{P}(X_A^t|{\bf X}_C^{[T]}, {\bf X}_B^{[t]})}{\mathds{P}(X_A^t| {\bf X}_C^{[T]}, {\bf X}_B^{[t-1]})}\right)\right] \right) \\ 
 &\phantom{=}  +  \beta_{A}^{t_A} \left(\mathds{E}\left[ \log\left(\frac{\mathds{P}(X_A^{t_A}|{\bf X}_C^{[T]}, {\bf X}_B^{[t_A]})}{\mathds{P}(X_A^{t_A}|{\bf X}_C^{[T]}, {\bf X}_B^{[t_A-1]})}\right)\right]  - \mathds{E}\left[ \log\left(\frac{\mathds{P}(g_A^{t_A|t_A-1}({\bf X}_A^{[t_A-1]})|{\bf X}_C^{[T]}, {\bf X}_B^{[t_A]})}{\mathds{P}(g_A^{t_A|t_A-1}({\bf X}_A^{[t_A-1]})|{\bf X}_C^{[T]}, {\bf X}_B^{[t_A-1]})}\right)\right] \right) \\ 
 &\phantom{=} + \sum_{t=t_A+1}^T \beta_A^t \left(\mathds{E}\left[ \log\left(\frac{\mathds{P}(g_A^{t|t_A}({\bf X}_A^{[t_A]})|{\bf X}_C^{[T]}, {\bf X}_B^{[t]})}{\mathds{P}(g_A^{t|t_A}({\bf X}_A^{[t_A]})|{\bf X}_C^{[T]}, {\bf X}_B^{[t-1]})}\right)\right] - \mathds{E}\left[ \log\left(\frac{\mathds{P}(g_A^{t|t_A-1}( {\bf X}_A^{[t_A-1]})| {\bf X}_C^{[T]},  {\bf X}_B^{[t]})}{\mathds{P}(g_A^{t|t_A-1}( {\bf X}_A^{[t_A-1]})|{\bf X}_C^{[T]}, {\bf X}_B^{[t-1]})}\right)\right] \right) \\ 
  &=   \beta_A^{t_A} \left(I\left({X}_B^{t_A}; {X}_A^{t_A} | {\bf X}_C^{[T]}, {\bf X}_B^{[t_A-1]}\right) - I\left({X}_B^{t_A}; g_A^{t_A|t_A-1}({\bf X}_A^{[t_A-1]}) | {\bf X}_C^{[T]}, {\bf X}_B^{[t_A-1]}\right) \right)  \\
  &\phantom{=}  +  \sum_{t=t_A+1}^{T} \beta_A^t \left( \mathds{E}\left[\log\frac{\mathds{P}({X}_B^t| g_A^{t|t_A}({\bf X}_A^{[t_A]}), {\bf X}_C^{[T]}, {\bf X}_B^{[t-1]})}{\mathds{P}({X}_B^t | {\bf X}_C^{[T]}, {\bf X}_B^{[t-1]})} \right]  -  \mathds{E}\left[\log\frac{\mathds{P}({X}_B^t| g_A^{t|t_A-1}({\bf X}_A^{[t_A-1]}), {\bf X}_C^{[T]}, {\bf X}_B^{[t-1]})}{\mathds{P}({X}_B^t | {\bf X}_C^{[T]}, {\bf X}_B^{[t-1]})} \right]\right) \\
  &=   \beta_A^{t_A} \left(I\left({X}_B^{t_A}; {X}_A^{t_A} | {\bf X}_C^{[T]}, {\bf X}_B^{[t_A-1]}\right) - I\left({X}_B^{t_A}; g_A^{t_A|t_A-1}({\bf X}_A^{[t_A-1]}) | {\bf X}_C^{[T]}, {\bf X}_B^{[t_A-1]}\right) \right)  \\
  &\phantom{=} +  \sum_{t=t_A+1}^{T} \beta_A^t\left( \mathds{E}_{ {\bf X}_A^{[t_A]}, {\bf X}_B^{[t]}, {\bf X}_C^{[T]}}\left[\log\frac{\mathds{P}({X}_B^t| g_A^{t|t_A}({\bf X}_A^{[t_A]}), {\bf X}_C^{[T]}, {\bf X}_B^{[t-1]})}{\mathds{P}({X}_B^t| g_A^{t|t_A-1}({\bf X}_A^{[t_A-1]}), {\bf X}_C^{[T]}, {\bf X}_B^{[t-1]})} \right] \right) \\
  &=   \beta_A^{t_A} \left(I\left({X}_B^{t_A}; {X}_A^{t_A} | {\bf X}_C^{[T]}, {\bf X}_B^{[t_A-1]}\right) - I\left({X}_B^{t_A}; g_A^{t_A|t_A-1}({\bf X}_A^{[t_A-1]}) | {\bf X}_C^{[T]}, {\bf X}_B^{[t_A-1]}\right) \right)  \\
  &\phantom{=}  +  \sum_{t=t_A+1}^{T} \beta_A^t\left( \mathds{E}_{{\bf X}_A^{[t_A]}, {\bf X}_B^{[t-1]}, {\bf X}_C^{[T]}}\left[\text{KL}\left({\mathds{P}({X}_B^t| g_A^{t|t_A}({\bf X}_A^{[t_A]}), {\bf X}_C^{[T]}, {\bf X}_B^{[t-1]}) || \mathds{P}({X}_B^t| g_A^{t|t_A-1}({\bf X}_A^{[t_A-1]}), {\bf X}_C^{[T]}, {\bf X}_B^{[t-1]})}\right) \right] \right)
 \end{split}
\end{equation}

Defining $\Delta R_{A^t}^\text{ target}(t_A) = 0$ for $t < t_A$,  $\Delta R_{A^{t_A}}^\text{ target}(t_A) =  I\left({X}_B^{t_A}; {X}_A^{t_A} | {\bf X}_C^{[T]}, {\bf X}_B^{[t_A-1]}\right) - I\left({X}_B^{t_A}; g_A^{t_A|t_A-1}({\bf X}_A^{[t_A-1]}) | {\bf X}_C^{[T]}, {\bf X}_B^{[t_A-1]}\right) $, and  for $t > t_A$ as $\Delta R_{A^t}^\text{target} =  \mathds{E}\left[\text{KL}\left({\mathds{P}({X}_B^t| g_A^{t|t_A}({\bf X}_A^{[t_A]}), {\bf X}_C^{[T]}, {\bf X}_B^{[t-1]}) || \mathds{P}({X}_B^t| g_A^{t|t_A-1}({\bf X}_A^{[t_A-1]}), {\bf X}_C^{[T]}, {\bf X}_B^{[t-1]})}\right) \right]$, we get $\Delta R_A^{\text{target}}(t_A) = \sum_{t=1}^T \beta_A^t \Delta R_{A^t}^\text{target}(t_A)$. Thus, we see that the marginal expected payment is:
\begin{equation}\small \label{margrew}
\begin{split}
    R_A(t_A) - R_A(t_A-1) &= \Delta R_A({t_A})  \\
    &= \Delta {R}_A^\text{expert}(t_A) + \Delta{R}_A^\text{target}(t_A) \\
    &= \sum_{t=1}^T \alpha_A^t \Delta R_{A^t}^\text{expert}(t_A) + \beta_A^t\Delta R_{A^t}^\text{target}(t_A) 
    \end{split}
\end{equation}

Finally we show non-negativity of marginal expected expert and target payments $\Delta R_{k}^\text{expert}, \Delta R_{k}^\text{target} \geq 0$, and positivity on criterion $t_A$ when the marginal criterion completed is $t_A$, i.e., $\Delta {R}_{A^{t_A}}^\text{expert}(t_A) > 0$ and $\Delta {R}_{A^{t_A}}^\text{target}(t_A) > 0$. 

We easily get the non-negativity of the marginal expected expert payment $\Delta R_{A^t}^\text{expert}(t_A) \geq 0$ from the non-negativity of mutual information. Now consider the marginal expected expert payment on criterion $t_A$ in particular ($\Delta {R}_{A^{t_A}}^\text{expert}(t_A)$). By second order stochastic relevance (Assumption \ref{stoch}), for any $x_A^{t_A} \neq \tilde{x}_A^{t_A}$ there is some $x_C^{t_A} \subset {\bf x}_C^{[t_A]}$ such that $\mathds{P}(X_B^{t_A}|x_A^{t_A}, {\bf x}_A^{[t_A-1]}, {\bf x}_C^{[t_A]}) \neq \mathds{P}(X_B^{t_A}|\tilde{x}_A^{t_A}, {\bf x}_A^{[t_A-1]}, {\bf x}_C^{[t_A]}) \Rightarrow  X_B^{t_A} \nCI X_A^{t_A} \mid {\bf X}_A^{[t_A-1]}, {\bf X}_C^{[t_A]}$. Thus, $I\left(X_B^{t_A}; X_A^{t_A} | {\bf X}_A^{[t_A-1]}, {\bf X}_C^{[t_A]}\right) > 0$. This gives us that $\Delta {R}_{A^{t_A}}^\text{expert}(t_A) > 0$ for all $t_A \in [T]$.

Next, we consider the marginal expected target payment.  By the non-negativity of mutual information and KL divergence, we have that $\Delta R_{A^t}^\text{target}(t_A) \geq 0$. We show that the inequality is strict when $t=t_A$, i.e.,
 $\Delta R_{A^{t_A}}^\text{target}(t_A) = I\left({X}_B^{t_A}; {X}_A^{t_A} | {\bf X}_C^{[T]}, {\bf X}_B^{[t_A-1]}\right) - I\left({X}_B^{t_A}; g_A^{t_A|t_A-1}({\bf X}_A^{[t_A-1]}) | {\bf X}_C^{[T]}, {\bf X}_B^{[t_A-1]}\right)  > 0$ as follows: 

\begin{equation} \small
\begin{split}
    I\left({X}_B^{t_A}; g_A^{t_A|t_A-1}({\bf X}_A^{[t_A-1]}) | {\bf X}_C^{[T]}, {\bf X}_B^{[t_A-1]}\right) &\leq I\left({X}_B^{t_A}; {\bf X}_A^{[t_A-1]} | {\bf X}_C^{[T]}, {\bf X}_B^{[t_A-1]}\right) \\
    &= I\left({X}_B^{t_A}; {\bf X}_A^{[T]} | {\bf X}_C^{[T]}, {\bf X}_B^{[t_A-1]}\right) -  I\left({ X}_B^{t_A}; {\bf X}_A^{[t_A, T]} | {\bf X}_A^{[t_A-1]}, {\bf X}_C^{[T]}, {\bf X}_B^{[t_A-1]}\right) \\
    &= I\left({X}_B^{t_A}; {X}_A^{t_A} | {\bf X}_C^{[T]}, {\bf X}_B^{[t_A-1]}\right) -  I\left({X}_B^{t_A}; {\bf X}_A^{[t_A, T]} | {\bf X}_A^{[t_A-1]}, {\bf X}_C^{[T]}, {\bf X}_B^{[t_A-1]}\right) \\
    &< I\left({X}_B^{t_A}; {X}_A^{t_A} | {\bf X}_C^{[T]}, {\bf X}_B^{[t_A-1]}\right)
\end{split}
\end{equation}

For the first inequality, we used the information monotonicity of mutual information. For the equality in the second step, we used the chain rule of mutual information. For the equality in the third step, we use the conditional independence assumption (Assumption \ref{condindep}). For the final inequality, we used the fact that second-order stochastic relevance (Assumption \ref{stoch}) implies that for any $x_A^{t_A} \neq \tilde{x}_A^{t_A}$ there is some $x_C^{t_A} \subset {\bf x}_C^{[T]}$ such that $\mathds{P}(X_B^{t_A}|x_A^{t_A}, {\bf x}_A^{[t_A-1]}, {\bf x}_C^{[T]}, {\bf x}_B^{[t_A-1]}) \neq \mathds{P}(X_B^{t_A}|\tilde{x}_A^{t_A}, {\bf x}_A^{[t_A-1]}, {\bf x}_C^{[T]}, {\bf x}_B^{[t_A-1]}) \Rightarrow  X_B^{t_A} \nCI X_A^{t_A} \mid {\bf X}_A^{[t_A-1]}, {\bf X}_C^{[T]}, {\bf X}_B^{[t_A-1]} \Rightarrow I\left(X_B^{t_A}; {\bf X}_A^{[t_A, T]} | {\bf X}_A^{[t_A-1]}, {\bf X}_C^{[T]}, {\bf X}_B^{[t_A-a]}\right) > 0$. 

The above derivation applies similarly for reviewer $B$ when peers $A, C$ complete all criteria report truthfully, and for reviewer $C$ when peers $A, B$ complete all criteria and report truthfully.
\end{proof}

\paragraph{Theorem \ref{optparam}}
\emph{Suppose mechanism hyperparameters $\{  \boldsymbol{\alpha}_A^{[T]}, \boldsymbol{\beta}_A^{[T]}, \boldsymbol{\alpha}_B^{[T]}, \boldsymbol{\beta}_B^{[T]}, \boldsymbol{\alpha}_C^{[T]}, \boldsymbol{\beta}_C^{[T]}\}$ are solutions to the following linear program:}
\begin{equation}\small
\begin{split}
    &\min \sum_{k \in \{A, B, C \}} \sum_{t=1}^T \alpha_k^t R_{k^t}^\text{expert}(T) +  \beta_k^t R_{k^t}^\text{target}(T)\\
    &\text{s.t. } \sum_{t=1}^T \alpha_A^t \Delta R_{A^t}^\text{expert}(t_A) + \beta_A^t\Delta R_{A^t}^\text{target}(t_A) > \Delta e_A(t_A) \text{ for } 1 \leq t_A \leq T \\
    &\phantom{\text{s.t. }} \sum_{t=1}^T \alpha_B^t \Delta R_{B^t}^\text{expert}(t_B) + \beta_B^t\Delta R_{B^t}^\text{target}(t_B) > \Delta e_B(t_B) \text{ for } 1 \leq t_B \leq T \\
    &\phantom{\text{s.t. }} \sum_{t=1}^T \alpha_C^t \Delta R_{C^t}^\text{expert}(t_C) + \beta_C^t\Delta R_{C^t}^\text{target}(t_C) > \Delta e_C(t_C) \text{ for } 1 \leq t_C \leq T \\
    &\phantom{\text{s.t. }} \alpha_k^t, \beta_k^t \geq 0
\end{split}
\end{equation}

\emph{Then, for any reviewer $k$, if their peers complete all criteria and report predictions and review scores truthfully, reviewer $k$ strictly maximizes their expected utility by also completing all criteria and reporting predictions and review scores truthfully. In other words, completing all criteria and reporting truthfully is a strict Bayesian-Nash equilibrium. We refer to this equilibrium as the fully-informative equilibrium.}

\begin{proof}
We require the hyperparameters to be non-negative so reviewers are paymented and not penalized for maximizing their expert and target payments. When reviewer $k$'s peers are completing all criteria and reporting truthfully, we want the mechanism to pay enough to cover reviewer $k$'s cost of effort and incentivize them to complete all criteria and report truthfully as well. By Theorem \ref{truththm}, we know that truthful reporting on shared completed criteria is a strict BNE, so when a reviewer $k$'s peers complete all criteria and report truthfully, every additional criterion completed by reviewer $k$ becomes a shared complete criterion and their payoff is maximized by truthfully reporting the review score on the marginal criterion. Thus, the question is simply how to set the hyperparameters to incentivize completion of all criteria (when peers complete all criteria and report truthfully). We analyse reviewer $A$'s incentives when reviewers $B, C$ complete all criteria and report truthfully, and this applies symmetrically to reviewers $B, C$. 

Reviewer $A$'s expected utility for completing $t_A$ criteria and reporting signals truthfully is $U_A = R_A(t_A) - e_A(t_A)$, so to ensure that the mechanism always covers the cost of effort, we want the marginal utility of completing an additional criterion $t_A$ to be positive, i.e., $R_A(t_A) - e_A(t_A) > R_A(t_A - 1) - e_A(t_A-1)$ for all $1 \leq t_A \leq T$. Rearranging this, we can write $R_A(t_A) - R_A(t_A - 1) > e_A(t_A) - e_A(t_A-1) \Rightarrow \Delta R_A(t_A) > \Delta e_A(t_A)$. $\Delta e_A(t_A)$ is the marginal cost of completing criterion $t_A$, and by Equation \ref{margrew}, the marginal expected payment for completing criterion $t_A$ when peers complete all criteria is $\Delta R_A(t_A) =  \sum_{t=1}^T \alpha_A^t \Delta R_{A^t}^\text{expert}(t_A) + \beta_A^t\Delta R_{A^t}^\text{target}(t_A)$. Thus, if the hyperparameters $\{\boldsymbol{\alpha}_A^{[T]}, \boldsymbol{\beta}_A^{[T]}\}$ satisfy $\sum_{t=1}^T \alpha_A^t \Delta R_{A^t}^\text{expert}(t_A) + \beta_A^t\Delta R_{A^t}^\text{target}(t_A) > \Delta e_A(t_A)$ for all $1 \leq t_A \leq T$, then reviewer $A$ maximizes their utility by completing all criteria and reporting truthfully. Satisfying these constraints is feasible since $ \Delta R_{A^t}^\text{expert}(t_A), R_{A^t}^\text{target}(t_A) > 0$ for $t=t_A$ as shown in Lemma \ref{margrewlemma}.

By the same argument, we desire $\{\boldsymbol{\alpha}_B^{[T]}, \boldsymbol{\beta}_B^{[T]}\}$ satisfy $\Delta R_B(t_B) > \Delta e_B(t_B)$, and $\{\boldsymbol{\alpha}_C^{[T]}, \boldsymbol{\beta}_C^{[T]}\}$ satisfy $\Delta R_C(t_C) > \Delta e_C(t_C)$.  Then, for any reviewer $k$, if their peers complete all criteria and report truthfully, their marginal utility of completing an additional criterion is always positive, so reviewer $k$ strictly maximizes their expected utility by also completing all criteria and reporting truthfully. Hence, completing all criteria and reporting truthfully is a strict Bayesian-Nash equilibrium.

In setting the hyperparameters, we also wish to minimize the mechanism's cost for having a paper reviewed, so we wish to minimize the total payout $R_A(T) + R_B(T) + R_C(T) = \min \sum_{k \in \{A, B, C \}} \sum_{t=1}^T \alpha_k^t R_{k^t}^\text{expert}(T) +  \beta_k^t R_{k^t}^\text{target}(T)$. The optimal choice of hyperparameters is found by minimizing the mechanism's cost under the aforementioned constraints by solving the given linear program where the constraints are modified to a weak inequality, and adding an $\epsilon$ to the solution of the LP to strictly satisfy the constraints.

\end{proof}

\paragraph{Lemma \ref{revwelf}}
\emph{Let $s_k = \left(e_k(t_k), \hat{\bf x}_k^{[T]}, \hat{\bf p}^+_{k'^{[T]} \leftarrow k}, \hat{\bf p}_{k'^{[T]} \leftarrow k}\right)$ be reviewer $k$'s strategy, where $k'$ is reviewer $k$'s target. If ${\bf s} = (s_A, s_B, s_C)$ is an equilibrium strategy profile,
 reviewer $A$'s expected payment is:}
 \begin{equation} \small
     \mathds{E}[r_A] = \sum_{t=1}^T \alpha_A^t\left( -\mathds{E}\left[H\left(\mathds{Q}\left({\hat{X}_B^t | {\bf x}_A^{[t_A]}, \hat{\bf x}_C^{[t]}}\right)\right)\right] -\mathds{E}\left[H\left(\mathds{Q}\left({\hat{X}_B^t | {\bf x}_A^{[t_A]}, \hat{\bf x}_C^{[t-1]}}\right)\right)\right]\right)   + \beta_A^t I\left(\hat{X}_B^t; \hat{ X}_A^t | {\bf X}_C^{[t_C]}, \hat{\bf X}_B^{[t-1]}\right)
 \end{equation}
 
\emph{(expected payment for peers $B,C$ computed analogously). Reviewers' total expected payment is:}
\begin{equation}\small
    \begin{split}
    \mathds{E}[r_A + r_B + r_C] &=  \sum_{t=1}^T \alpha_A^t\left( -\mathds{E}\left[H\left(\mathds{Q}\left({\hat{X}_B^t | {\bf x}_A^{[t_A]}, \hat{\bf x}_C^{[t]}}\right)\right)\right] -\mathds{E}\left[H\left(\mathds{Q}\left({\hat{X}_B^t | {\bf x}_A^{[t_A]}, \hat{\bf x}_C^{[t-1]}}\right)\right)\right]\right)   + \beta_A^t I\left(\hat{X}_B^t; \hat{ X}_A^t | {\bf X}_C^{[t_C]}, \hat{\bf X}_B^{[t-1]}\right) \\
        &\phantom{=} + \alpha_B^t\left( -\mathds{E}\left[H\left(\mathds{Q}\left({\hat{X}_C^t | {\bf x}_B^{[t_B]}, \hat{\bf x}_A^{[t]}}\right)\right)\right] -\mathds{E}\left[H\left(\mathds{Q}\left({\hat{X}_C^t | {\bf x}_B^{[t_B]}, \hat{\bf x}_A^{[t-1]}}\right)\right)\right]\right)   + \beta_B^t I\left(\hat{X}_C^t; \hat{ X}_B^t | {\bf X}_A^{[t_A]}, \hat{\bf X}_C^{[t-1]}\right) \\
        &\phantom{=} + \alpha_C^t\left( -\mathds{E}\left[H\left(\mathds{Q}\left({\hat{X}_A^t | {\bf x}_C^{[t_C]}, \hat{\bf x}_B^{[t]}}\right)\right)\right] -\mathds{E}\left[H\left(\mathds{Q}\left({\hat{X}_A^t | {\bf x}_C^{[t_A]}, \hat{\bf x}_B^{[t-1]}}\right)\right)\right]\right)   + \beta_C^t I\left(\hat{X}_A^t; \hat{ X}_C^t | {\bf X}_B^{[t_B]}, \hat{\bf X}_A^{[t-1]}\right)
    \end{split}
\end{equation}

\begin{proof}[Proof of Lemma \ref{revwelf}]

 First, we observe that since ${\bf s}$ is an equilibrium strategy profile, all reviewers' \emph{predictions} are truthful, i.e., $\hat{\bf p}^+_{B^t \leftarrow A} = \mathds{Q}\left(\hat{X}_B^t | {\bf X}_A^{[t_A]}, \hat{\bf X}_C^{[t]} \right) $ 
 and $\hat{\bf p}_{B^t \leftarrow A} = \mathds{Q}\left(\hat{X}_B^t | {\bf X}_A^{[t_A]}, \hat{\bf X}_C^{[t-1]} \right)$, and so on for the other reviewers' expert predictions. Predictions are truthful in any equilibrium as otherwise reviewers could unilaterally improve their expert payment under the strictly proper log scoring rule by deviating to a truthful prediction. The expected payment for reviewers (expectation of Equation \ref{rew}) for criterion $t$ is:
\begin{equation}\label{allrew}
\begin{split}
    \mathds{E}_{\mathds{Q}}[r_{A^t}] &=  \mathds{E}[ \alpha_A^t\left(S_\text{log}(\hat{x}_B^t, \hat{\bf p}^+_{B^t \leftarrow A}) + S_\text{log}(\hat{x}_B^t, \hat{\bf p}_{B^t \leftarrow A})\right) + \beta_A^t\left(S_\text{log}(\hat{x}_A^t, \hat{\bf p}^+_{A^t \leftarrow C}) - S_\text{log}(\hat{x}_A^t, \hat{\bf p}_{A^t \leftarrow C})\right)] \\
    \mathds{E}_{\mathds{Q}}[r_{B^t}] &=   \mathds{E}[ \alpha_A^t\left(S_\text{log}(\hat{x}_C^t, \hat{\bf p}^+_{C^t \leftarrow B}) + S_\text{log}(\hat{x}_C^t, \hat{\bf p}_{C^t \leftarrow B})\right) + \beta_B^t\left(S_\text{log}(\hat{x}_B^t, \hat{\bf p}^+_{B^t \leftarrow A}) - S_\text{log}(\hat{x}_B^t, \hat{\bf p}_{B^t \leftarrow A})\right) ]\\ 
  \mathds{E}_{\mathds{Q}}[r_{C^t}] &=  \mathds{E}[ \alpha_C^t\left(S_\text{log}(\hat{x}_A^t, \hat{\bf p}^+_{A^t \leftarrow C}) + S_\text{log}(\hat{x}_A^t, \hat{\bf p}_{A^t \leftarrow C})\right) + \beta_C^t\left(S_\text{log}(\hat{x}_C^t, \hat{\bf p}^+_{C^t \leftarrow B}) - S_\text{log}(\hat{x}_C^t, \hat{\bf p}_{C^t \leftarrow B})\right)]
    \end{split}
\end{equation}

 Now, observe that reviewer $A$'s expected expert payment (first two terms in expectation) on some criterion $t$ can be written as follows:

\begin{equation}\small
\begin{split}
    \alpha_A^t\mathds{E}[S_\text{log}(\hat{x}_B^t, \hat{\bf p}^+_{B^t \leftarrow A}) + S_\text{log}(\hat{x}_B^t, \hat{\bf p}_{B^t \leftarrow A})] &= \alpha_A^t\mathds{E}\left[\mathds{Q}\left(\hat{X}_B^t | {\bf X}_A^{[t_A]}, \hat{\bf X}_C^{[t]} \right) + \mathds{Q}\left(\hat{X}_B^t | {\bf X}_A^{[t_A]}, \hat{\bf X}_C^{[t-1]} \right) \right] \\
    &= -\mathds{E}_{{\bf X}_A^{[t_A]}, \hat{\bf X}_C^{[t]}}\left[H\left(\mathds{Q}\left({\hat{X}_B^t | {\bf x}_A^{[t_A]}, \hat{\bf x}_C^{[t]}}\right)\right)\right] -\mathds{E}_{ {\bf X}_A^{[t_A]}, \hat{\bf X}_C^{[t-1]}}\left[H\left(\mathds{Q}\left({\hat{X}_B^t | {\bf x}_A^{[t_A]}, \hat{\bf x}_C^{[t-1]}}\right)\right)\right] 
    \end{split}
\end{equation}

where $H\left(\cdot \right)$ is the Shannon entropy. 
For reviewer $A$'s target payment (last two terms in expectation), we can use Equation \ref{logscore} to write:
\begin{equation}\label{eq:condmut}
    \mathds{E}_{\mathds{Q}}[ S_\text{log}(\hat{x}_A^t, \hat{\bf p}^+_{A^t \leftarrow C}) - S_\text{log}(\hat{x}_A^t, \hat{\bf p}_{A^t \leftarrow C})] = I\left(\hat{X}_B^t; \hat{ X}_A^t | {\bf X}_C^{[t_C]}, \hat{\bf X}_B^{[t-1]}\right)
\end{equation}
\begin{remark} 
Each reviewer's target payment on criterion $t$ is exactly the HMIP payment for that criterion under the conditional independence assumption (Assumption 4.3 in \citet{kong2018eliciting}). 
\end{remark}

We can then write reviewer $A$'s total expected payment on criterion $t$ as follows (and the same applies to $B$ and $C$):
\begin{equation}\label{eq:singrew}\small 
    \begin{split}
        \mathds{E}_{\mathds{Q}}[r_{A^t}]
        &= \alpha_A^t\left( -\mathds{E}_{{\bf X}_A^{[t_A]}, \hat{\bf X}_C^{[t]}}\left[H\left(\mathds{Q}\left({\hat{X}_B^t | {\bf x}_A^{[t_A]}, \hat{\bf x}_C^{[t]}}\right)\right)\right] -\mathds{E}_{ {\bf X}_A^{[t_A]}, \hat{\bf X}_C^{[t-1]}}\left[H\left(\mathds{Q}\left({\hat{X}_B^t | {\bf x}_A^{[t_A]}, \hat{\bf x}_C^{[t-1]}}\right)\right)\right]\right)   + \beta_A^t I\left(\hat{X}_B^t; \hat{ X}_A^t | {\bf X}_C^{[t_C]}, \hat{\bf X}_B^{[t-1]}\right)
    \end{split}
\end{equation}

Making the analogous computation for reviewers $B$ and $C$,  we can compute the total expected payment for the reviewers on a single criterion $t$:
\begin{equation}\small
    \begin{split}
        \mathds{E}_{\mathds{Q}}[r_{A^t} + r_{B^t} + r_{C^t}]
        &= \alpha_A^t\left( -\mathds{E}\left[H\left(\mathds{Q}\left({\hat{X}_B^t | {\bf x}_A^{[t_A]}, \hat{\bf x}_C^{[t]}}\right)\right)\right] -\mathds{E}\left[H\left(\mathds{Q}\left({\hat{X}_B^t | {\bf x}_A^{[t_A]}, \hat{\bf x}_C^{[t-1]}}\right)\right)\right]\right)   + \beta_A^t I\left(\hat{X}_B^t; \hat{ X}_A^t | {\bf X}_C^{[t_C]}, \hat{\bf X}_B^{[t-1]}\right) \\
        &\phantom{=} + \alpha_B^t\left( -\mathds{E}\left[H\left(\mathds{Q}\left({\hat{X}_C^t | {\bf x}_B^{[t_B]}, \hat{\bf x}_A^{[t]}}\right)\right)\right] -\mathds{E}\left[H\left(\mathds{Q}\left({\hat{X}_C^t | {\bf x}_B^{[t_B]}, \hat{\bf x}_A^{[t-1]}}\right)\right)\right]\right)   + \beta_B^t I\left(\hat{X}_C^t; \hat{ X}_B^t | {\bf X}_A^{[t_A]}, \hat{\bf X}_C^{[t-1]}\right) \\
        &\phantom{=} + \alpha_C^t\left( -\mathds{E}\left[H\left(\mathds{Q}\left({\hat{X}_A^t | {\bf x}_C^{[t_C]}, \hat{\bf x}_B^{[t]}}\right)\right)\right] -\mathds{E}\left[H\left(\mathds{Q}\left({\hat{X}_A^t | {\bf x}_C^{[t_A]}, \hat{\bf x}_B^{[t-1]}}\right)\right)\right]\right)   + \beta_C^t I\left(\hat{X}_A^t; \hat{ X}_C^t | {\bf X}_B^{[t_B]}, \hat{\bf X}_A^{[t-1]}\right)
    \end{split}
\end{equation}

Thus, the total expected payment for the reviewers for all criteria is:
\begin{equation}\small
    \begin{split}
    \mathds{E}_{\mathds{Q}}[r_A + r_B + r_C] &= \mathds{E}\left[  \sum_{k \in \{A, B, C\}} \sum_{t=1}^T  r_{k^t} \right] \\
    &= \sum_{t=1}^T \alpha_A^t\left( -\mathds{E}\left[H\left(\mathds{Q}\left({\hat{X}_B^t | {\bf x}_A^{[t_A]}, \hat{\bf x}_C^{[t]}}\right)\right)\right] -\mathds{E}\left[H\left(\mathds{Q}\left({\hat{X}_B^t | {\bf x}_A^{[t_A]}, \hat{\bf x}_C^{[t-1]}}\right)\right)\right]\right)   + \beta_A^t I\left(\hat{X}_B^t; \hat{ X}_A^t | {\bf X}_C^{[t_C]}, \hat{\bf X}_B^{[t-1]}\right) \\
        &\phantom{=} + \alpha_B^t\left( -\mathds{E}\left[H\left(\mathds{Q}\left({\hat{X}_C^t | {\bf x}_B^{[t_B]}, \hat{\bf x}_A^{[t]}}\right)\right)\right] -\mathds{E}\left[H\left(\mathds{Q}\left({\hat{X}_C^t | {\bf x}_B^{[t_B]}, \hat{\bf x}_A^{[t-1]}}\right)\right)\right]\right)   + \beta_B^t I\left(\hat{X}_C^t; \hat{ X}_B^t | {\bf X}_A^{[t_A]}, \hat{\bf X}_C^{[t-1]}\right) \\
        &\phantom{=} + \alpha_C^t\left( -\mathds{E}\left[H\left(\mathds{Q}\left({\hat{X}_A^t | {\bf x}_C^{[t_C]}, \hat{\bf x}_B^{[t]}}\right)\right)\right] -\mathds{E}\left[H\left(\mathds{Q}\left({\hat{X}_A^t | {\bf x}_C^{[t_A]}, \hat{\bf x}_B^{[t-1]}}\right)\right)\right]\right)   + \beta_C^t I\left(\hat{X}_A^t; \hat{ X}_C^t | {\bf X}_B^{[t_B]}, \hat{\bf X}_A^{[t-1]}\right)
    \end{split}
\end{equation}
\end{proof}

\paragraph{Theorem \ref{ir}}
 \begin{enumerate}[label=(\alph*)]
     \item \emph{In the highest paying uninformative equilibrium (where all reviewers exert zero effort), every reviewer's expected payment (and hence aggregate reviewer welfare) is zero.}
     
     \item \emph{Suppose that in addition to the constraints specified in Theorem \ref{optparam}, the mechanism hyperparameters also satisfy $R_k(T) - e_k(T) > 0 \Leftrightarrow \sum_{t=1}^T\alpha_k^t R_{k^t}^\text{expert}(T) + \beta_k^t R_{k^t}^\text{target}(T)  > e_k(T)$ for $k \in \{A, B, C\}$, i.e., reviewers' expected utility in the fully-informative equilibrium is positive. Then, it is individually rational to participate in the mechanism and complete all criteria and report truthfully, given peers are doing the same. Additionally, the fully-informative equilibrium has higher individual and aggregate utility than the uninformed equilibrium.}
 \end{enumerate}

 \begin{proof}[Proof of Theorem \ref{ir}]
 \phantom{a}
 
 \begin{enumerate}[label=(\alph*)]
 
 \item When all reviewers exert zero effort, their reports are not sampled from the common prior. Instead, they may report some fixed a priori most likely sequence of signals. In such an equilibrium, expert predictions without the source will be perfect and the updated predictions will be the same, as the expert can simply predict the constant signals with complete certainty. Thus, the entropy of expert predictions is minimized and expert payments will be 0. Since expert predictions will not change, all target payments will also be 0 in expectation. Intuitively, this is because the signals do not come from the common prior, and do not contain information about other reviewers' signals.
 
 In any other uninformative equilibrium, expected target payments will still be zero since mutual information of uninformed signals not drawn from the common prior will be zero. However, the expert payment can only decrease (the maximum possible expert payment is zero). Thus, in the highest paying uninformative equilibrium, reviewers individually and in aggregate are paid nothing in expectation.
 
 \item It is individually rational for a reviewer to participate in the mechanism when their expected utility is positive. We want individual rationality under our desired fully informative equilibria, so we achieve positive expected utility by requiring $R_k(T) - e_k(T) > 0$ or equivalently $\sum_{t=1}^T \alpha_k^t  R_{k^t}^\text{expert}(T) + \beta_k^t R_{k^t}^\text{target}(T) > e_k(T) $.
 
 Denoting the expected utility of some arbitrary uninformative equilibrium as $R_k(0,0 ,0)$, we observe that $R_k(T) > e_k(T) > 0 \geq R_k(0,0,0)$, thus the individual utility of the fully-informative equilibrium is higher than any uninformative equilibrium. This holds for all reviewers, so the aggregate utility of the fully-informative equilibrium is also higher than the uninformative equilibrium.
 \end{enumerate}
 \end{proof}

\newpage

\end{document}